\documentclass[a4paper,USenglish,cleveref, autoref, thm-restate]{lipics-v2021}
\usepackage{booktabs} 
\usepackage{comment}
\usepackage{cite}
\usepackage{glossaries}
\usepackage{float}
\glsdisablehyper
\newacronym{pbs}{PBS}{Proposer-Builder Separation}
\newacronym{ofc}{OFC}{Order Flow Composition}
\newacronym{bft}{BFT}{Byzantine Fault Tolerant}
\newacronym{pos}{PoS}{Proof-of-Stake}
\newacronym{pow}{PoW}{Proof-of-Work}
\newacronym{mev}{MEV}{Maximal Extractable Value}
\newacronym[shortplural=DEXes]{dex}{DEX}{Decentralized Exchange}
\newacronym[shortplural=CEXes]{cex}{CEX}{Centralized Exchange} 
\newacronym{amm}{AMM}{Automated Market Maker}
\newacronym{defi}{DeFi}{Decentralized Finance}
\newacronym{tradfi}{TradFi}{Traditional Finance}
\newacronym{pga}{PGA}{Priority Gas Auction}
\newacronym{tob}{ToB}{Top-of-Block}
\newacronym{bob}{BoB}{Body-of-Block}
\newacronym{eob}{EoB}{End-of-Block}
\newacronym{ofa}{OFA}{Order Flow Auction}
\newacronym{0xce91228789b57deb45e66ca10ff648385fe7093b}{0xCe91228789B57DEb45e66Ca10Ff648385fE7093b}{MEV Blocker Rebates Safe}
\newacronym{rfq}{RFQ}{Request For Quote}
\newacronym{pm}{PM}{Profit Margin}
\newacronym{eoa}{EOA}{Externally Owned Account}
\newacronym{eof}{EOF}{Exclusive Order Flow}
\newacronym{msof}{MSOF}{Most Significant Order Flow}
\newacronym{jit}{JIT}{Just-In-Time}
\newacronym{lda}{LDA}{Linear Discriminant Analysis}
\newacronym{da}{DA}{Decoding Accuracy}
\newacronym{rpc}{RPC}{Remote Procedure Call}
\newacronym{brt}{BRT}{\texttt{beaverbuild}, \texttt{rysnc}, and \texttt{Titan}}
\newacronym{ep}{EP}{Exclusive Provider}
\newacronym{tee}{TEE}{Trusted Execution Environment}
\newacronym{il}{IL}{Inclusion List}
\newacronym{aps}{APS}{Attestor-Proposer Separation}
\newacronym{ea}{EA}{Execution Auction}
\newacronym{cr}{CR}{Censorship-Resistance}
\newacronym{hft}{HFT}{High-Frequency Trading}
\newacronym{l2}{L2}{Layer-2}
\newacronym{l1}{L1}{Layer-1}
\newacronym{bsc}{BSC}{Binance Smart Chain}
\newacronym{zk}{ZK}{Zero Knowledge}
\newacronym{mcp}{MCP}{Multiple Concurrent Proposers}
\newacronym{gcp}{GCP}{Google Cloud Platform}
\newacronym{tvl}{TVL}{Total Value Locked}
\newacronym{sia}{SIA}{Sequence-Independent Arbitrage}
\newacronym{sda}{SDA}{Sequence-Dependent Arbitrage}
\newacronym{cdf}{CDF}{Cumulative Distribution Function}
\newacronym{ci}{CI}{Confidence Interval}
\newacronym{fp}{FP}{False Positive}
\newacronym{fn}{FN}{False Negative}
\newacronym{clob}{CLOB}{Central Limit Order Book}
\newacronym{p2p}{P2P}{Peer-to-Peer}
\newacronym{eip}{EIP}{Ethereum Improvement Proposal}
\newacronym{abm}{ABM}{Agent-Based Modeling}
\newacronym{fdv}{FDV}{Fully Diluted Valuation}
\newacronym{cpmm}{CPMM}{Constant Product Market Maker}
\newacronym{fcfs}{FCFS}{First-Come, First-Served}
\newacronym{epbs}{ePBS}{enshrined Proposer-Builder Separation}
\newacronym{se}{SE}{Simulation Experiment}
\newacronym{exp}{EXP}{Experiment}
\newacronym{hhi}{HHI}{Herfindahl–Hirschman Index}
\newacronym{poa}{PoA}{Price of Anarchy}
\newcommand{\encircled}[2][0.9mm]{%
    \raisebox{.8pt}{%
        \textcircled{%
            \raisebox{0.35pt}{%
                \kern #1
                \scalebox{0.70}{#2}
            }%
        }%
    }%
}
\makeatletter
\newcommand*{\ensquared}[1]{\relax\ifmmode\mathpalette\@ensquared@math{#1}\else\@ensquared{#1}\fi}
\newcommand*{\@ensquared@math}[2]{\@ensquared{$\m@th#1#2$}}
\newcommand*{\@ensquared}[1]{%
\tikz[baseline,anchor=base]{\node[draw,outer sep=0pt,inner sep=0.6mm,minimum width=3.8mm] {#1};}} 
\makeatother

\hideLIPIcs  %uncomment to remove references to LIPIcs series (logo, DOI, ...), e.g. when preparing a pre-final version to be uploaded to arXiv or another public repository

\makeatletter
\@ifundefined{parhead}
  {\newcommand{\parhead}[1]{\par\noindent\textbf{#1.}}}
  {\renewcommand{\parhead}[1]{\par\noindent\textbf{#1.}}}
\makeatother
\title{The Price of Decentralization in Block Building}

\author{Burak {\"O}z\textsuperscript{*}}%
  {Flashbots, Germany}{burak@flashbots.net}{https://orcid.org/0009-0003-7508-7112}{}
  
\author{Fei Wu\textsuperscript{*\dag}}%
  {King's College London, United Kingdom}%
  {fei.wu@kcl.ac.uk}{https://orcid.org/0009-0004-5717-0219}{}

\author{Luis Correia}{Flashbots, United Kingdom}{luis@flashbots.net}{https://orcid.org/0009-0001-6772-7437}{}

\author{Sen Yang}{Yale University, United States \and IC3, United States}{sen.yang@yale.edu}{https://orcid.org/0000-0002-8866-2097}{}

\author{Bruno Mazorra}{Flashbots, Spain}{bruno@flashbots.net}{https://orcid.org/0000-0003-0779-0765}{}

\author{Stefanos Leonardos}{King's College London, United Kingdom}{stefanos.leonardos@kcl.ac.uk}{https://orcid.org/0000-0002-1498-1490}{}

\authorrunning{Burak {\"O}z et al.} %TODO mandatory. First: Use abbreviated first/middle names. Second (only in severe cases): Use first author plus 'et al.'

\Copyright{Burak {\"O}z, Fei Wu, Luis Correia, Sen Yang, Bruno Mazorra, and Stefanos Leonardos} %TODO mandatory, please use full first names. LIPIcs license is "CC-BY";  http://creativecommons.org/licenses/by/3.0/

\ccsdesc[500]{Theory of computation~Algorithmic game theory}
\ccsdesc[500]{Computing methodologies~Modeling and simulation}

\keywords{decentralization, block building, potential game, price of anarchy} %TODO mandatory; please add comma-separated list of keywords

\category{} %optional, e.g. invited paper

\supplement{}

\supplementdetails[subcategory={Source Code}, cite={}, swhid={}]{Software}{https://github.com/boz1/db-sims}

\nolinenumbers %uncomment to disable line numbering

\EventEditors{Aggelos Kiayias and Maria Kyropoulou}
\EventNoEds{2}
\EventLongTitle{8th Conference on Advances in Financial Technologies (AFT 2026)}
\EventShortTitle{AFT 2026}
\EventAcronym{AFT}
\EventYear{2026}
\EventDate{October 6--9, 2026}
\EventLocation{London, UK}
\EventLogo{}
\SeriesVolume{395}
\ArticleNo{27}
\newcommand{\added}[1]{\textcolor{blue}{#1}}
\newcommand{\deleted}[1]{\textcolor{red}{\sout{#1}}}
\newcommand{\revised}[2]{\textcolor{red}{\sout{#1}}\textcolor{blue}{#2}}
\renewcommand{\added}[1]{#1}
\renewcommand{\deleted}[1]{}
\renewcommand{\revised}[2]{#2}

\newif\iffullversion

\fullversiontrue % full arXiv version

\iffullversion
    \newcommand{\fullref}[1]{\Cref{#1}}
\else
    \newcommand{\fullref}[1]{the appendix of the full version~\cite{fullver}}
\fi

\begin{document}

\maketitle

\let\oldthefootnote\thefootnote
\let\thefootnote\relax
\footnotetext{\textsuperscript{*}These authors contributed equally to this work.}
\footnotetext{\textsuperscript{\dag}Fei Wu performed work in part during an internship at Flashbots in Spring 2026.}
\let\thefootnote\oldthefootnote

\begin{abstract}
Decentralized block building mechanisms replace the monopoly of a single proposer with multiple builders. However, their censorship-resistance and fair-access benefits depend not only on the number of builders, but also on where builders are geographically positioned to provide transaction
coverage. We study this tension between builder location choice, user transaction coverage, and utility concentration by modeling decentralized block building as a stochastic coverage game. Builders choose regions, information sources emit transactions over a block construction round, and latency determines whether a transaction is received before the deadline.

We show that the game is an exact potential game and admits a pure Nash equilibrium. We prove an asymptotically tight factor-$2$ bound on the Price of Anarchy, which we interpret as the \emph{price of decentralization} from uncoordinated builder placement. We also study builder utility concentration, showing that the lowest-utility builder earns at least half of the highest-utility builder's payoff, and the utility-share HHI is at most \(12.5\%\) above the egalitarian benchmark. We complement the theory with
simulations under richer latency and source environments and show that slot times, builder participation, and reward-sharing rules are important protocol design choices that can shape the price of decentralization.
\end{abstract}

\section{Introduction}
Many blockchain protocols assign the block proposal right to a single leader in each slot.
%, such as the validator selected to propose the next block in Ethereum.
This temporary monopoly gives the proposer control over transaction inclusion and ordering, creating opportunities for \gls{mev}~\cite{daian2019flashboys20frontrunning}.
% This proposer has the unilateral right to decide which transactions to include or exclude and how to order them, creating opportunities for economic value extraction through strategic manipulation, referred to as \gls{mev}~\cite{daian2019flashboys20frontrunning}. Although temporary, this monopolistic power can threaten fair access for users, especially for time-sensitive transactions. For instance, on-chain auctions require timely and credible bid inclusion: a single proposer can unilaterally censor, delay, or selectively include bids, thereby affecting the auction outcome~\cite {fox_et_al:LIPIcs.AFT.2023.19}.
%applications such as on-chain auctions cannot be implemented without additional privacy techniques because they are exposed to strategic censorship by the proposer~\cite{fox_et_al:LIPIcs.AFT.2023.19}. 
%Even if the proposer does not exercise this power for economic gain, 
Moreover, proposers located in particular jurisdictions may face censorship pressures, such as those induced by OFAC sanctions in the U.S.~\cite{wahrstatter_blockchain_2023}, and geographic concentration can create potential liveness risks in the presence of regional outages. Recent multiple-proposer or decentralized block-building proposals seek to reduce this monopoly and improve censorship resistance by allowing multiple builders to construct candidate blocks or sub-blocks~\cite{pranav2025mcp,kniep_solana_constellation_2026}.\footnote{In the paper, we use \emph{builder} as the generic term for an entity that constructs and proposes a candidate (sub-)block.} If a user transaction can be included by multiple builders, who typically share the fees paid by the transaction, this increases the cost of censorship: excluding a transaction may require bribing multiple entities rather than a single proposer. 

However, multiple builders do not automatically imply geographically broad transaction coverage. In permissionless blockchains, the consensus protocol typically does not mandate where builders should operate. Because latency affects whether transactions reach builders before the block construction deadline, self-interested builders may prefer regions that provide access to the same high-value transaction sources. Recent work shows that Ethereum's block-building protocol design is not geographically neutral and can create co-locating incentives to minimize latency to payoff-critical parties~\cite{yang_geographical_2025}. If builders concentrate in a small number of regions, the benefits of having multiple builders can diminish, as users located farther away may face weaker inclusion guarantees due to latency effects. \added{In auction-like settings, such latency differences can give colocated parties an informational advantage~\cite{moallemi2026latencyadvantagescommonvalueauctions}.}

% if builders concentrate in a small number of regions, the benefits of having multiple builders can diminish as users located farther away may face weaker inclusion guarantees due to latency effects. In auction-like settings, such latency differences can give colocated parties an informational advantage~\cite{moallemi2026latencyadvantagescommonvalueauctions}. Thus, the geographic placement of builders can play a critical role in determining whether such decentralized block building mechanisms enable fair economic participation from geographically distributed users. Recent work~\cite{yang_geographical_2025} shows that Ethereum's block-building protocol design is not geographically neutral and can create co-locating incentives to minimize latency to payoff-critical parties.

% In permissionless blockchains, the consensus protocol typically does not mandate where block-building nodes should be run. Instead, builders are operated by self-interested players who choose regions to maximize their utility. Such incentives may lead to inefficient equilibrium outcomes in which no builder wants to unilaterally deviate, even though the resulting placement is socially suboptimal. In particular, builders may duplicate access to the same valuable transaction sources rather than collectively expanding coverage across regions. Such outcomes can undermine the intended censorship-resistance benefits of decentralized block building and may also lead to welfare losses and concentration of economic activity.

This raises a critical design question: if decentralized block building mechanisms rely on multiple self-interested builders who choose locations to maximize individual utility, will their location choices actually improve access for geographically distributed users and transaction sources? How inefficient and concentrated can the resulting placements be? These questions matter because the censorship-resistance and fair-access benefits of multi-builder designs depend not only on the number of builders, but also on whether their locations are distributed enough to provide timely coverage to users and transaction sources across regions over the existing Internet infrastructure with latency constraints.

In this paper, we model decentralized block building as a stochastic coverage game (cf. \Cref{sec:model}). Builders choose geographic regions, while information sources
%, representing generators of transaction value, 
emit transactions over the course of a block-construction round. A builder's region determines the probability that a transaction emitted by a given source reaches the builder before the block deadline. Since multiple builders may receive and include the same transaction, decentralized placement can create redundant coverage: several builders may cover the same high-value sources, while other sources remain weakly covered. 
% This modeling choice emphasizes protocol-level transaction coverage rather than endogenous user routing or builder pricing; we discuss alternative location-market formulations in \Cref{sec:discussion}.

We study the welfare loss induced by this decentralized builder region game (cf. \Cref{sec:theory}). Relative to a social planner who coordinates builder locations to maximize expected transaction coverage, selfish builders may over-concentrate in regions that provide access to the same valuable information sources. We measure this efficiency loss through the \gls{poa}, which we interpret in this setting as the \emph{price of decentralization}, and derive an asymptotically tight worst-case bound on it using the smoothness framework~\cite{roughgarden_intrinsic_nodate}. The key observation is that when a transaction is already covered by several builders, an additional builder receives only a reduced share of its value under the equal-split sharing rule, while not improving welfare. This allows us to relate unilateral deviations to aggregate welfare. We also prove the existence of a pure Nash equilibrium by constructing an exact potential function, and establish bounds on utility concentration among builders at equilibrium. In our model, welfare captures the total transaction value covered by at least one builder, while utility concentration captures how the economic benefits of block building are distributed.

The theoretical results characterize worst-case properties: whether stable placements exist, how inefficient they can be, and how concentrated builder utilities can become. To complement this analysis, we run simulations in richer latency and source environments. These simulations show when welfare losses arise in concrete instances, how they are driven by under-coverage of peripheral sources, and why geographic concentration and builder utility concentration need not coincide (cf. \Cref{sec:simulation_setup,sec:results}). 
The results highlight how protocol designs can shape welfare in decentralized block building (cf. \Cref{sec:discussion}).
We summarize our contributions as follows.
\begin{itemize}
% [leftmargin=*,topsep=2pt,itemsep=1pt]
    \item We model decentralized block building as a latency-induced stochastic coverage game in which builders choose regions and transaction sources generate value over time. 
    %This formulation connects geographic placement, probabilistic transaction reception, and redundant coverage in a single game-theoretic model.
    We analyze the equilibrium and efficiency properties of the game under the equal-split sharing rule. We show that the game admits a pure Nash equilibrium through an exact potential function and prove a tight factor-$2$ PoA bound. 
    %verify that its stochastic coverage structure preserves the valid utility game conditions needed for the standard factor-$2$ PoA bound, and show that this bound is asymptotically tight.
    \item We study utility concentration among builders in equilibrium. We show that the lowest-utility builder earns at least half of the highest-utility builder's payoff, and that the resulting utility-share HHI is at most \(12.5\%\) above the egalitarian benchmark. Both bounds are tight.
    \item We use simulations~\cite{price-of-dececentralization-sim} to study the game under richer latency and source environments. We find that welfare losses are largest in intermediate regimes where peripheral sources are reachable, but selfish incentives still favor high-value source regions. We also show that geographic concentration and utility concentration need not align: planner allocations may improve coverage by placing builders in lower-payoff peripheral regions, while equilibrium outcomes can be more geographically centralized but more utility-balanced.
    % \item We connect our findings to protocol design, highlighting how timing and builder participation can shape the price of decentralization in block building, and outlining extensions to endogenous transaction-routing models and alternative reward-sharing rules.
\end{itemize}

\subsection{Related work}

\parhead{Geographic (de)centralization} 
Prior measurement studies show that Ethereum nodes are geographically concentrated, with deployments repeatedly skewed toward a small number of countries, particularly in Europe and North America~\cite{kim2018measuring,ccaf2025ethereum,kiffer2025multiple,ethernodes2025countries,heimbach2025deanonymizing}. 
These studies establish geographic concentration as an empirical phenomenon, while recent work studies the incentives and mechanisms that shape geographic decentralization.
Yang et al.~\cite{yang_geographical_2025} show that Ethereum's block building regimes create location-dependent validator rewards, allowing latency advantages and asymmetric access to information sources to reinforce concentration.
Roeschlin et al.~\cite{roeschlin2026incentivizing} study how to enforce geography-aware incentives when nodes may misreport their locations, using RTT-based localization to support truthful reporting.
Our work is complementary to both lines. 
We study equilibrium outcomes in decentralized block building when builders strategically choose geographic regions, and analyze how latency-sensitive placement affects welfare, geographic concentration, and builder utility concentration.

\parhead{Decentralized block building} Decentralized block building aims to reduce centralized control over transaction inclusion and ordering, especially under MEV-driven incentives, and to improve censorship resistance.
Industry discussions have advocated distributed, parallel, and global block building architectures to counter latency moats and censorship risks~\cite{flashbots_decentralized_building_2026,miller2024parallel}.
Recent work proposes \gls{mcp} designs to reduce the serial monopoly of a single proposer over inclusion and ordering~\cite{pranav2025mcp,kniep_solana_constellation_2026,pranav2025tfmmcp}. 
Complementary work studies transaction assignment in such protocols, showing that stronger censorship resistance often requires more transaction duplication and hence lower throughput~\cite{elsheimy2026censorship}. 
Our work complements these works through placement equilibria: we study how builders' strategic geographic choices affect welfare, geographic concentration, and builder utility concentration.

\parhead{Coverage games and equilibrium inefficiency} Vetta introduced valid utility games, where agents' utilities are tied to a monotone submodular social objective and Nash equilibria admit constant welfare guarantees~\cite{vetta2002utilitygames}. Gairing studied covering games, where players select subsets of weighted elements and covered value is allocated through utility-sharing rules~\cite{gairing_covering_2009}. 
\revised{Our equal-split rule corresponds to symmetric equal sharing: each transaction's value is divided equally among the builders that include it, regardless of source-specific cardinalities. Unlike prior work on abstract coverage objectives and sharing-rule design, we study equal sharing in a stochastic block building model with latency-induced transaction visibility, proving equilibrium existence, tight efficiency bounds, and utility-concentration bounds.}{Our equal-split rule corresponds to symmetric equal sharing, but coverage is stochastic and induced by heterogeneous source-to-builder reception probabilities.}

\added{Prior work studies the distributional properties of equilibria through notions
such as utility uniformity~\cite{utilityuniformity} and fairness ratio for costs~\cite{fairnessratio}. Our utility dispersion and HHI results complement this literature by studying reward concentration in a stochastic coverage game.}

\added{Beyond worst-case \gls{poa} guarantees, related work examines inefficiency in particular parameter regimes and among outcomes reached in practice or through specified response dynamics~\cite{asymptoticpoa,empiricalpoa,dynamicpoa}. Our simulations take a similar instance-level perspective by evaluating equilibria selected by asynchronous better-response dynamics across sampled environments and initial conditions.}

\section{Model}\label{sec:model}

\subsection{Sources, regions, and probabilistic reception}
We study a decentralized block-building network in which, in every round of a blockchain protocol, a new block containing transactions is appended. Each round has a duration $\Delta>0$. Denote the set of builders by $\mathcal{B}:=\{1,\dots,K\}$ and the set of geographical regions in the network by $\mathcal{R}:=\{1,\dots,R\}$.

In each round, every builder constructs a sub-block from the transactions it receives before the round deadline. We assume that all received transactions are included and impose no block-capacity constraint. We discuss the implications of block-capacity constraints in~\Cref{sec:discussion}.

Let \(\mathcal{I}\) be a finite set of information sources. Each source \(I\in\mathcal{I}\) is associated with a fixed region \(r(I)\in\mathcal{R}\). For each source $I$, let $\mathcal{J}_I$ denote the random set of transactions emitted by source $I$ during one round. Each transaction $j \in \mathcal{J}_I$ has an emission time $t_j \in [0, \Delta]$ and a positive value $V_j$.

We assume that source $I$ has finite expected emitted value in a round. Specifically, there exist a scalar $\Lambda_I>0$ and a density $\rho_I:[0,\Delta]\to\mathbb{R}_+$ with $\int_0^\Delta \rho_I(t)\,dt = 1$ such that, for every bounded function $h:[0,\Delta]\to\mathbb{R}$,
\[
\mathbb{E}\!\left[\sum_{j \in \mathcal{J}_I} V_j\, h(t_j)\right]
=
\Lambda_I \int_0^\Delta h(t)\rho_I(t)\,dt.
\]
Thus, $\Lambda_I$ is the expected total value emitted by source $I$ in one round, and $\rho_I$ describes how this expected value is distributed over time within this round. This formulation lets us work directly with expected value flow rather than specifying the full arrival process of transactions and values.

For each ordered region pair $(i,k)\in\mathcal{R}\times\mathcal{R}$, let $D_{k\to i}$ denote the propagation latency from a source in region $k$ to a builder located in region $i$, and let $F_{k\to i}$ be its cumulative distribution function. We allow the latency distribution to depend on the ordered region pair, so different source-builder pairs may exhibit different propagation characteristics. This captures persistent geographical differences in network connectivity while allowing for random latency fluctuations.

A strategy for builder $b$ consists of selecting a region $s_b\in\mathcal{R}$, and a strategy profile is written as \(\mathbf{s}:=(s_1,\dots,s_K)\in\mathcal{R}^K\). Region selection is modeled as a one-shot simultaneous-move game: builders choose regions based on the known source environment, but not on the contemporaneous choices of the other builders.

Consider a transaction emitted at time $t\in[0,\Delta]$ by a source located in region $k$. If builder $b$ is located in region $i=s_b$, then the transaction reaches builder $b$ before the block deadline with probability
\[
q_k(i,t)
:=
\Pr(D_{k\to i}\leq \Delta-t)
=
F_{k\to i}(\Delta-t).
\]
Thus, for a source $I$ and a builder $b$, the corresponding reception probability is obtained by setting $k=r(I)$ and $i=s_b$, giving $q_{r(I)}(s_b,t)$. Since reception probabilities depend on builders only through their chosen regions, all builders located in the same region have identical reception probabilities.

Let $\mathcal{T}_b$ denote the random set of transactions received by builder $b$ in a round, and write \(\mathcal{T}_b^I := \mathcal{T}_b \cap \mathcal{J}_I\) for the transactions from source $I$ received by $b$. Conditional on the realized transactions and their emission times, receptions are independent across builders. Thus, two builders located in the same region have the same reception probability for a given transaction, but need not have the same realized reception outcome.

\subsection{Value sharing rule and welfare}
After a round, the sub-blocks of all builders are aggregated into a final block. When multiple builders include the same transaction, its value is shared among them. For each transaction $j$, let
\[
n_j(\mathbf{s}) := \left|\{b\in\mathcal{B}: j\in\mathcal{T}_b\}\right|
\]
denote the number of builders that received and included transaction $j$ under profile $\mathbf{s}$. When the profile is clear from context, we write $n_j$.

We assume an \emph{equal-split sharing rule}: if $n_j\ge 1$, then each builder that included transaction $j$ receives the fraction $1/n_j$ of its value. Builders that did not include $j$ receive zero from it. Equivalently, the reward assigned to builder $b$ from transaction $j$ is
\[
V_j \cdot \frac{\mathbf{1}\{j\in\mathcal{T}_b\}}{n_j},
\]
with the convention that this term is zero when $n_j=0$.

We use equal splitting as a simple symmetric and budget-balanced sharing rule: it depends only on the number of builders that include the transaction, treats all including builders equally, and distributes exactly \(V_j\) when the transaction is covered by at least one builder. Thus, the total reward distributed for transaction \(j\) is \(V_j\mathbf{1}\{n_j\ge 1\}\). This rule is consistent with Shapley-style distribution designs for multi-proposer transaction fee mechanisms, where the reward from a transaction is shared among the proposers that include it~\cite{pranav2025tfmmcp}. We abstract from the bidding layer and treat \(V_j\) as the transaction payment to be shared; alternative reward-sharing rules are discussed in~\Cref{sec:discussion}.

The realized reward of builder $b$ under profile $\mathbf{s}$ is
\[
\Pi_b(\mathbf{s})
=
\sum_{I\in\mathcal{I}}\sum_{j\in\mathcal{J}_I}
V_j \cdot \frac{\mathbf{1}\{j\in\mathcal{T}_b\}}{n_j},
\]
again with the convention that the summand is zero when $n_j=0$.

Since transaction generation, including values and emission times, and reception outcomes are random, $\Pi_b(\mathbf{s})$ is a random variable. The payoff-relevant quantity for builder $b$ is its expected reward
\[
    u_b(\mathbf{s}) := \mathbb{E}[\Pi_b(\mathbf{s})],
\]
where the expectation is taken over transaction generation and reception outcomes. Given the other builders' region choices, each builder chooses $s_b$ to maximize its expected reward.

We define the \emph{social welfare} of profile $\mathbf{s}$ as the expected total value of transactions received by at least one builder:
\[
W(\mathbf{s})
:=
\mathbb{E}\!\left[
\sum_{I\in\mathcal{I}}\sum_{j\in\mathcal{J}_I}
V_j\,\mathbf{1}\{n_j\ge 1\}\right].
\]
Thus, $W(\mathbf{s})$ measures the total user value realized by transactions that are successfully covered under profile $\mathbf{s}$. Following the standard welfare view in prior work~\cite{tfm_post_mev}, we treat transfers between users and builders as redistributing value rather than creating it. Therefore, in our stylized model, welfare equals the total value of included transactions.

The equal-split sharing rule implies a useful accounting identity: aggregate builder utility equals welfare. We state this formally in \Cref{lem:aggregate_utility}.

\subsection{Equilibrium and social optimum}
The preceding definitions induce a finite game in which builders are the players, regions are the actions, and expected rewards \(u_b\) are the payoffs. We refer to this game as the \emph{builder region game}.

A pure strategy profile $\mathbf{s}$ is a Nash equilibrium of the builder region game if no builder can improve its expected reward through a unilateral deviation. That is, for every $b\in\mathcal{B}$ and every $s_b' \in \mathcal{R}$,
\[
u_b(\mathbf{s}) \ge u_b(s_b',\mathbf{s}_{-b}),
\]
where $\mathbf{s}_{-b}$ denotes the strategies of all builders other than $b$.

A social planner chooses the entire profile $\mathbf{s}$ jointly to maximize welfare. We denote a welfare-maximizing profile by
\[
\mathbf{s}^* \in \arg\max_{\mathbf{s}\in\mathcal{R}^K} W(\mathbf{s}),
\]
and call $W(\mathbf{s}^*)$ the optimal welfare.
\section{Equilibrium, Efficiency, and Utility Concentration}\label{sec:theory}
In this section, we analyze the equilibrium, efficiency, and utility-concentration properties of the builder region game. We first show that the equal-split sharing rule induces an exact potential function, which implies the existence of a pure Nash equilibrium. We then bound the welfare loss in equilibrium using the \gls{poa}. Finally, we study concentration in builder utilities through the \gls{hhi}. We provide the proofs in~\fullref{app:proofs}.

\subsection{Pure Nash equilibrium via an exact potential}\label{subsec:potential}
We first show that the builder region game admits a pure Nash equilibrium by constructing an exact potential function. Under equal sharing, welfare itself is not a potential function: if a transaction is already covered by other builders, an additional builder can still obtain positive utility by also receiving it, while welfare does not increase. The key observation is instead that the marginal payoff from joining a transaction already received by $m$ other builders is exactly $V_j/(m+1)$, which corresponds to the increment of a harmonic potential.

\begin{proposition}[Exact potential under equal sharing]\label{prop:potential}
Consider the builder region game under the equal-split sharing rule. For $m\geq 1$, let
\[
H_m := \sum_{\ell=1}^m \frac{1}{\ell},
\qquad
H_0 := 0
\]
denote the harmonic numbers. For a realized transaction $j$, let 
\[n_j(\mathbf{s}) := \left|\{b \in \mathcal{B} : j \in \mathcal{T}_b\}\right|\]
denote the number of builders that receive j under profile $\mathbf{s}$. Define
\[
\Phi(\mathbf{s}) := \mathbb{E}\left[\sum_{I \in \mathcal{I}}
    \sum_{j \in \mathcal{J}_I}V_j \, H_{n_j(\mathbf{s})}\right].
\]
Then $\Phi$ is an exact potential for the expected-utility game. That is, for every builder $b \in \mathcal{B}$, every profile $\mathbf{s}$, and every deviation $s_b' \in \mathcal{R}$,
\[
\Phi(s_b', \mathbf{s}_{-b}) - \Phi(\mathbf{s})
=
u_b(s_b', \mathbf{s}_{-b}) - u_b(\mathbf{s}).
\]
Hence, under equal sharing, the builder region game is a finite exact potential game.
\end{proposition}

% \begin{corollary}[Existence of pure Nash equilibrium]\label{cor:pure_ne}
% The builder region game under the equal-split sharing rule admits a pure Nash equilibrium.
% \end{corollary}

% \begin{proof}
% By \Cref{prop:potential}, the game is a finite exact potential game. Hence, any maximizer of $\Phi$ over $\mathcal{R}^K$ is a pure Nash equilibrium~\cite{monderer_potential_1996}.
% \end{proof}

% \begin{remark}[Asynchronous better-response dynamics]\label{rem:abr_pne}
% The potential-game structure also implies finite-time convergence of asynchronous better-response (ABR) dynamics. Starting from any profile, suppose that in each round a single builder is selected and moves to a region that yields strictly higher expected utility than its current region, if such a region exists. By \Cref{prop:potential}, every strict unilateral payoff improvement strictly increases the potential $\Phi$. Since the profile space $\mathcal{R}^K$ is finite, this process can contain only finitely many strict improvements and must terminate at a profile with no profitable unilateral deviation, i.e., a pure Nash equilibrium.
% \sen{the profile space should be consistent. do we need to explain that it is $\binom{|\mathcal R|+K-1}{K}$ here or in the simulation?} \fei{No. Here $\mathcal{R}^K$ is still the correct pure strategy profile space, which means it applies to any game setting. Only after we initialize the game as a symmetric game in the simulation, the computational shortcut of the profile numbers work.}
% \end{remark}
\begin{corollary}[Pure equilibria and better-response convergence]
\label{cor:pure_ne}
Under the equal-split sharing rule, the builder region game admits a pure Nash equilibrium. Moreover, every asynchronous sequence of strict better responses terminates in finite time at a pure Nash equilibrium.
\end{corollary}

\begin{proof}
By \Cref{prop:potential}, the game is a finite exact potential game. Hence, any maximizer of \(\Phi\) over the finite pure profile space \(\mathcal R^K\) is a pure Nash equilibrium~\cite{monderer_potential_1996}.

Now consider an asynchronous sequence of strictly better responses. At each step, a single builder changes its region in a way that strictly increases its expected utility. By exact potentiality, this strictly increases the potential \(\Phi\). Since the pure profile space \(\mathcal R^K\) is finite, such strict improvements cannot continue indefinitely. The process must therefore terminate at a profile with no profitable unilateral deviation, i.e., a pure Nash equilibrium.
\end{proof}

\added{In the deterministic-reception special case, the game admits a direct congestion-game representation, where transaction sources are resources and a source of value $\Lambda_I$ gives payoff $\Lambda_I/m$ to each of its $m$ covering builders, who dilute each others rewards. With probabilistic reception, expected payoff depends on the full vector of heterogeneous probabilities, so this compact source-level representation is less direct. Abstractly, every finite exact potential game can be mapped to a congestion game \cite{monderer_potential_1996}}. 

\subsection{Price of Anarchy in decentralized block building}\label{subsec:poa}
In the builder region game, builders may over-concentrate in regions that provide access to valuable information sources. This can lead multiple builders to include the same transactions, creating redundant coverage, while other sources remain weakly covered. The resulting mismatch between private incentives and aggregate coverage can reduce welfare. We use \gls{poa} to measure the welfare loss caused by selfish region selection relative to the welfare-maximizing profile chosen by a social planner.

Let $\mathcal{E}$ denote the set of pure Nash equilibria of the builder region game, and let $\mathbf{s}^*$ be a welfare-maximizing profile. We define \gls{poa} as
\[
    \mathrm{PoA}
    :=
    \frac{W(\mathbf{s}^*)}
    {\displaystyle\min_{\mathbf{s}\in\mathcal{E}} W(\mathbf{s})}.
    \label{eq:poa}
\]

The bound we derive follows the standard \gls{poa} argument for valid-utility games~\cite{vetta2002utilitygames}, a class that includes coverage and location-game variants, expressed through the smoothness framework~\cite{roughgarden_intrinsic_nodate}. The smoothness framework bounds inefficiency by comparing the total payoff players would obtain by unilaterally deviating from a current profile toward a benchmark profile with the welfare of the two profiles. Our contribution is to verify that the builder region game satisfies the required structure after taking expectations over transaction value flows and latency-induced probabilistic reception. Specifically, we show that welfare admits a source-level coverage representation, that equal-split sharing makes aggregate builder utility equal welfare, and that deviation payoffs can be lower-bounded by exclusive coverage. These properties allow the standard smoothness proof to carry over to our stochastic block-building setting.

A payoff-maximization game is $(\lambda,\mu)$-smooth if, for every pair of profiles $\mathbf{s}$ and $\mathbf{s}^*$,
\[
\sum_{b\in\mathcal{B}} u_b(s_b^*,\mathbf{s}_{-b})
\geq
\lambda W(\mathbf{s}^*)-\mu W(\mathbf{s}).
\]
If a game is $(\lambda,\mu)$-smooth and aggregate utility equals welfare, then every pure Nash equilibrium has welfare at least a $\lambda/(1+\mu)$ fraction of the optimal welfare. Equivalently,
\[
\mathrm{PoA}\leq \frac{1+\mu}{\lambda}.
\]

\begin{definition}[Coverage]\label{def:coverage}
A transaction from source $I$ emitted at time $t$ is \emph{covered} under profile $\mathbf{s}$ if at least one builder receives it before the deadline. The corresponding coverage probability is 
\[
f_I(t;\mathbf{s})
:=
1-\prod_{b\in\mathcal{B}}
\bigl(1-q_{r(I)}(s_b,t)\bigr).
\]
The expected coverage of source $I$ under profile $\mathbf{s}$, averaged with respect to the source's emission-time profile $\rho_I$, is
\[
\bar f_I(\mathbf{s})
:=
\int_0^\Delta \rho_I(t) f_I(t;\mathbf{s})\,dt.
\]
\end{definition}

\begin{lemma}[Coverage representation of welfare]\label{lem:welfare_coverage}
For every profile $\mathbf{s}$, total expected welfare satisfies
\[
W(\mathbf{s})
=
\sum_{I\in\mathcal{I}}
\Lambda_I \bar f_I(\mathbf{s}).
\]
\end{lemma}

The proof applies the expected value-flow assumption with $h(t)=f_I(t;\mathbf{s})$ and then sums over sources.

\noindent\textbf{Interpretation.}
The lemma shows that welfare depends only on whether each transaction is covered by at least one builder. Additional builders receiving the same transaction do not create additional social value. Thus, welfare losses arise from inefficient region allocation: if too many builders cover the same high-value sources, other sources may remain weakly covered relative to the social optimum.

\begin{lemma}[Aggregate utility identity]\label{lem:aggregate_utility}
Under the equal-split sharing rule, aggregate builder utility equals welfare. That is, for every profile $\mathbf{s}$,
\[
\sum_{b\in\mathcal{B}} u_b(\mathbf{s}) = W(\mathbf{s}).
\]
\end{lemma}
This follows because a covered transaction distributes exactly its full value across the builders that include it, while an uncovered transaction distributes zero.

\begin{lemma}[Smoothness inequality]\label{lem:smoothness}
The builder region game is $(1,1)$-smooth. That is, for all profiles $\mathbf{s}$ and $\mathbf{s}^*$,
\[
\sum_{b\in\mathcal{B}}u_b(s_b^*,\mathbf{s}_{-b})
\geq
W(\mathbf{s}^*) - W(\mathbf{s}).
\]
\end{lemma}
The proof lower-bounds each builder's deviation payoff by the value it would cover exclusively under the deviation. Summing these exclusive-coverage terms across builders and applying a union-bound argument shows that they cover the welfare gap \(W(\mathbf{s}^*)-W(\mathbf{s})\).

\begin{theorem}[Price of Anarchy bound]\label{thm:poa}
The \gls{poa} of the builder region game is at most $2$. Equivalently, for every pure Nash equilibrium $\mathbf{s}\in\mathcal{E}$ and every welfare-maximizing profile $\mathbf{s}^*$,
\[
W(\mathbf{s}) \geq \frac{1}{2}W(\mathbf{s}^*).
\]
\end{theorem}

\begin{remark}[Broader equilibrium notions]\label{rem:broader_eq}
The factor-$2$ guarantee follows from the standard smoothness argument for valid utility games~\cite{vetta2002utilitygames}. The block-building-specific content is the verification that a latency-induced stochastic region-selection game preserves this structure in expectation. Although \Cref{thm:poa} is stated for pure Nash equilibria, the same smoothness argument implies the factor-$2$ welfare guarantee for mixed Nash, correlated, and coarse correlated equilibria, with welfare evaluated in expectation~\cite{roughgarden_intrinsic_nodate}.
\end{remark}

We next construct an example to show that the factor-$2$ bound is asymptotically tight. The worst case in this game arises when all builders crowd into a region with a single highly valuable source that is already fully covered. A unilateral move to an uncovered secondary source's region would improve social coverage but reduce builder's utility. In contrast, a social planner would keep only one builder at the dominant source and allocate the remaining builders to uncovered source regions.

\begin{proposition}[Tightness of the factor-$2$ bound]\label{prop:tightness}
For every $K \geq 2$ and every $\varepsilon \in (0,1)$, there exists an instance of the builder region game with $K$ builders, such that some pure Nash equilibrium $\mathbf{s}$ satisfies
\[
\frac{W(\mathbf{s}^*)}{W(\mathbf{s})}
= 
2-\frac{1}{K}-\frac{K-1}{K}\varepsilon.
\]
In particular, both the \gls{poa} and PoS can be arbitrarily close to $2$.
\end{proposition}

\added{
\begin{remark}[Price of Stability]
\label{remark:pos}
The same construction also implies that the best and worst equilibria coincide, and the game admits a unique pure Nash equilibrium. Therefore, the factor-$2$ bound is asymptotically tight for the Price of Stability.  
\end{remark}
}

\subsection{Builder utility dispersion and concentration}\label{subsec:concentration}
% \bruno{Probably also true for MNE}

We next study how dispersed builder utilities can be at equilibrium. \added{Our
utility-dispersion measure coincides with the notion of \emph{utility uniformity} studied in~\cite{utilityuniformity}.} We then use this dispersion bound to quantify utility concentration among builders under decentralized block building with the equal-split sharing rule.

\begin{lemma}[Equilibrium utility dispersion]
\label{lem:utility_dispersion}
Under the equal-split sharing rule, let $\mathbf{s}$ be a pure Nash equilibrium. Define
\[u_{\min}(\mathbf{s}) := \min_{b \in \mathcal{B}} u_b(\mathbf{s}),
\qquad
u_{\max}(\mathbf{s}) := \max_{b \in \mathcal{B}} u_b(\mathbf{s}).
\]
Then
\[
u_{\min}(\mathbf{s})
\geq
\frac{1}{2}u_{\max}(\mathbf{s}).
\]
\end{lemma}

\begin{proposition}[Tightness of the utility dispersion bound]\label{prop:utility_dispersion_tight}
The factor $1/2$ in \Cref{lem:utility_dispersion} is tight: there exists an instance and a pure Nash equilibrium $\mathbf{s}$ such that
\[
u_{\min}(\mathbf{s}) = \frac{1}{2} u_{\max}(\mathbf{s}).
\]
\end{proposition}

\begin{definition}[Utility shares and concentration]\label{def:utility_hhi}
For any profile $\mathbf{s}$ with $\sum_{b \in \mathcal{B}} u_b(\mathbf{s}) > 0$, define the utility share of builder $b$ by
\[
x_b(\mathbf{s}) := \frac{u_b(\mathbf{s})}{\sum_{c \in \mathcal{B}} u_c(\mathbf{s})}.
\]
The Herfindahl--Hirschman Index (HHI) of the utility profile is
\[
\mathrm{HHI}(\mathbf{s}) := \sum_{b \in \mathcal{B}} x_b(\mathbf{s})^2.
\]
Its minimum possible value over $K$ builders is attained at equal utilities:
\[
\mathrm{HHI}^{\mathrm{egal}} = \frac{1}{K}.
\]
\end{definition}

\begin{proposition}[Equilibrium bound on utility concentration]\label{prop:utility_hhi}
Assume the equal-split sharing rule, and let $\mathbf{s}$ be a pure Nash equilibrium with $\sum_{b \in \mathcal{B}} u_b(\mathbf{s}) > 0$. Then
\[
\mathrm{HHI}(\mathbf{s}) \le \frac{9}{8K} = \frac{9}{8}\,\mathrm{HHI}^{\mathrm{egal}}.
\]
In addition, the factor $9/8$ is tight as a universal bound.
\end{proposition}

\noindent\textbf{Interpretation.}
This proposition shows that, under decentralized block building with the equal-split sharing rule, utility shares among builders cannot be much more concentrated than the egalitarian benchmark. At any pure Nash equilibrium, the utility-share \gls{hhi} is at most \(12.5\%\) above the value attained when all builders earn equal utility. This follows from the equilibrium dispersion bound: the lowest-utility builder must earn at least half of the highest-utility builder's payoff.
\section{Simulation Setup}\label{sec:simulation_setup}
The theoretical results provide worst-case guarantees that hold uniformly across model instances. Simulations allow us to ask a complementary, instance-level question: how do welfare, transaction coverage, and concentration vary with concrete source-value, latency, and builder-participation configurations? To answer this question, we compare self-interested equilibrium outcomes with centralized planner benchmarks across a set of controlled simulation environments. For reproducibility, we open-source our simulator~\cite{price-of-dececentralization-sim}. 

\parhead{Regions and information sources} We instantiate the region set $\mathcal{R}$ using 24 regions from \gls{gcp}, partitioned into two source-location pools (see \fullref{app:gcp-regions}).
The high-value pool contains regions near major financial or technology centers, such as North America, Western Europe, and East Asia\footnote{This classification follows recent observations that Ethereum latency competition depends on relay geography and price-signal paths across North America, Europe, and East Asia~\cite{dataalways2025geography}.}. The peripheral pool contains regions that are geographically further from these centers.
For each simulation instance, we sample the same number of source regions from the high-value and peripheral pools, without replacement, and place one information source in each selected region.
%first choose the locations of information sources by sampling the same number of regions from the high-value pool and the peripheral pool, without replacement.  We then place one source in each sampled region.These source locations remain fixed throughout the simulation run.
Source locations remain fixed throughout a run. We denote the high-value and peripheral source clusters by $\mathcal I_H$ and $\mathcal I_P$, respectively.
% We refer to sources placed in the high-value region pool as the high-value source cluster, denoted by $\mathcal I_H$, and sources placed in the peripheral region pool as the peripheral source cluster, denoted by $\mathcal I_P$.
We normalize the total expected emitted value to $\sum_{I\in\mathcal I}\Lambda_I=10$ in all reported experiments. Source-value asymmetry is parameterized by the high-to-low value share ratio $\gamma$, defined as the ratio between the total expected value emitted by $\mathcal I_H$ and $\mathcal I_P$.
% We parameterize source-value asymmetry by the high-to-low source value ratio $\rho$, defined as the ratio between the total value emitted by the high-value source cluster and the peripheral source cluster.
Thus, the two clusters receive total expected values 
$10\cdot\frac{\gamma}{1+\gamma}$ and $10\cdot\frac{1}{1+\gamma}$, respectively, divided evenly within each cluster.
% Thus, $\rho=1$ gives a $50$--$50$ split, while $\rho=10$ gives a $10/11$ versus $1/11$ split.

\parhead{Transaction and latency}
To generate realized transactions, the simulator uses a Poisson arrival process with emission times sampled uniformly in $[0, \Delta]$. In the reported experiments, transaction values are deterministic and chosen so that each source has the desired expected emitted value $\Lambda_I$. As in previous literature \cite{yang_geographical_2025,wu2024strategic}, the simulator also supports lognormal transaction values, but we set the value noise to zero to isolate latency and source-value effects. When sweeping \(\Delta\), we hold each $\Lambda_I$ fixed, so slot-duration experiments isolate the deadline effect rather than changing total emitted value per slot.

% For each source $I$, transactions arrive according to a Poisson process with rate $\lambda_I$. Over a slot of duration $\Delta$, the number of emitted transactions is $N_I \sim \mathrm{Poisson}(\lambda_I \Delta)$, with emission times sampled uniformly from $[0,\Delta]$. Transaction values are lognormal, $V_I \sim \mathrm{Lognormal}(\mu_I,\sigma_I)$.
% In the reported experiments, we use deterministic per-transaction values chosen so that each source has the desired expected emitted value $\Lambda_I$. The simulator also supports lognormal transaction values, but we set the noise to zero to isolate the effects of latency and source-value asymmetry.
% \done\fei{in the reported experiments, we use deterministic per-transaction values chosen so that each source has the desired expected emitted value $\Lambda_I$. The simulator also supports lognormal transaction values, but we set the noise to zero to isolate the effects of latency and source-value asymmetry.}

We calibrate the propagation delays using an empirical GCP pairwise latency matrix. For a source region $k$ and builder region $i$, latency $D_{k\to i}$ is modeled as lognormal with mean equal to the calibrated empirical latency and standard deviation 15\% of the mean, with a small lower bound to avoid degeneracy. The corresponding reception probability is $q_k(i,t)=\Pr(D_{k\to i}\le \Delta-t).$

\parhead{Equilibrium profiles} Builders choose regions from the same 24-region set and are initialized uniformly at random. We use asynchronous better-response dynamics (ABR), which is guaranteed to converge to a pure Nash equilibrium (PNE) under the exact-potential structure of the game (cf. \Cref{cor:pure_ne}). Builders update one at a time in round-robin order. When builder $b$ is selected, all other builders' regions are fixed; candidate regions are considered in random order, and $b$ moves to the first region that strictly improves its expected utility. If no improving region exists, it stays.\par

% \stef{Not sure if we want to show some robustness of findings in the appendix, e.g. by considering simultaneous updates as this asynchronous best-response involves some kind of coordination among builders (taking turns) that may not be present in practice - simultaneous methods should also often converge but may just take a bit longer (for policy gradient, we show this even for Markov potential games).} \fei{in most settings, it seems that we end up in a periodic cycle between 2 states with simultaneous update} 

In the simulator, utilities are evaluated numerically, and convergence is declared when a full round-robin sweep completes without any builder changing regions. 
In all reported experiments, every simulation run converged to a profile with no utility-improving unilateral deviation within 250 builder updates. We denote the equilibrium profile by $\mathbf{s}^{\mathrm{PNE}}$.

To compute a builder's utility at a candidate region, we use the equal-split sharing rule directly. If builder $b$ considers moving to region $i \in \mathcal{R}$, then for a transaction from source $I$ emitted at time $t$, builder $b$ receives the transaction with probability $q_{r(I)}(i,t)$. If $X^{-b}_{I,t}$ other builders also receive it, then builder $b$'s share is $\frac{1}{1+X^{-b}_{I,t}}$. Thus,
\[
    u_b(i,\mathbf{s}_{-b})
    =
    \sum_{I\in\mathcal I}
    \Lambda_I
    \int_0^\Delta
    \rho_I(t)\,
    q_{r(I)}(i,t)\,
    \mathbb{E}\!\left[\frac{1}{1+X^{-b}_{I,t}}\right]dt.
\]
Here, $X^{-b}_{I,t}$ is the random number of builders other than $b$ that receive a transaction from source $I$ emitted at time $t$, under the profile $\mathbf{s}_{-b}$. Since receptions are independent across builders, $X^{-b}_{I,t}$ follows a Poisson-binomial distribution with parameters $\{q_{r(I)}(s_{b'},t)\}_{b'\neq b}$. We compute the probability mass function by dynamic programming, evaluate $\mathbb{E}[1/(1+X^{-b}_{I,t})]$, and approximate the time integral using a uniform grid of points.
For each parameter setting, we sample five random source layouts. For each source layout, we simulate three independently sampled initial builder profiles, giving 15 runs per parameter setting. 

\subsection{Benchmarks and metrics}
\parhead{Planner benchmark} 
For each simulation instance, we compare the equilibrium profile with a centralized social-planner benchmark. 
% that jointly chooses the builder-location profile to maximize expected welfare. 
% Following the model definition, let $\mathbf{s}^* \in \arg\max_{\mathbf{s}\in\mathcal R^K} W(\mathbf{s})$ and write $W^*:=W(\mathbf{s}^*)$.
When exact optimization is tractable, the planner benchmark is $\mathbf{s}^* \in \arg\max_{\mathbf{s}\in\mathcal R^K} W(\mathbf{s})$, and $ W^*:=W(\mathbf{s}^*).$
We assume builders are symmetric, so exact search can be performed over occupancy vectors rather than over all labeled profiles in $\mathcal{R}^K$. The number of such occupancy vectors is $\binom{|\mathcal R|+K-1}{K}$. 
% When the planner problem is computationally tractable, we compute this benchmark by exhaustive search over the full combinatorial space of feasible $K$-builder placement multisets over the candidate region set $\mathcal R$.
% This space has size $\binom{|\mathcal R|+K-1}{K}$, which grows quickly with both $|\mathcal R|$ and $K$.
We use exact exhaustive search in the experiments with fixed number of builders $K=5$.
For experiments that sweep over larger builder counts, exact search becomes computationally expensive, so we use a greedy welfare-maximization routine as the planner benchmark. Across all tested instances, the worst observed greedy-to-exact welfare ratio is 0.999129 (see \fullref{app:greedy-validation}). 
% We compare the greedy benchmark with exhaustive search on tractable instances in \Cref{app:greedy-validation}.
In experiments using the greedy benchmark, welfare ratios should be interpreted relative to the greedy planner benchmark rather than the exact optimum.
% In all cases, the same planner routine is used consistently within a given experiment when computing welfare ratios and planner-side placement metrics.
% \done\fei{builders are symmetric, exhaustive search should not be $R^K$ but $(\frac{K+R-1}{K})$}

% Since exhaustive search over $\mathcal R^K$ is computationally expensive for our larger simulation instances, we approximate $W^*$ using a greedy welfare maximization routine.

\parhead{Metrics}
The primary efficiency metric is \emph{welfare ratio}, calculated as $\mathrm{WR}:=W_{\mathrm{PNE}}/\widehat W^*$, where $\widehat W^* = W^*$ for exact planner search and equals the greedy planner welfare when the greedy benchmark is used. 
When $\widehat W^*=W^*$, the theoretical $\mathrm{PoA}\le 2$ bound corresponds to $\mathrm{WR}\ge 1/2$.
% When the greedy planner is used, $\mathrm{WR}$ should instead be interpreted as welfare relative to the greedy planner benchmark.

% \parhead{Utility HHI} 
We report utility concentration with \emph{utility HHI} following \Cref{def:utility_hhi}. This metric measures how concentrated expected builder utilities are at the selected PNE.

Geographical concentration of builder locations is measured with \emph{geographic HHI}. Let $g_i(\mathbf{s}) := |\{b\in \mathcal{B} : s_b=i\}|/K$ be the fraction of builders located in region $i\in\mathcal R$. We define
\[
    \mathrm{HHI}_{\mathrm{geo}}(\mathbf{s})
    =
    \sum_{i\in\mathcal R} g_i(\mathbf{s})^2 .
\]
When $K\le |\mathcal R|$, its minimum value is $1/K$, attained when all builders occupy distinct regions. Its maximum value is 1, attained when all builders co-locate in a single region.
% \boz{after going over the results, I think our geo hhi metric is not very informative. In first three experiments, we have 5 symmetric high value source and 5 builders; it is very natural that we get low hhi. A more informative HHI definition could be based on concentration in high value sources vs periphery vs other regions. Or, having more builders.}\luis{I agree. In the previous paper there were 1k validators for 40 regions so geo hhi was informative, but thats not the case here. I think if we want to measure geographical concentration we should add a distance metric like normalised mean pairwise distance. Your idea of seeing whether builders are in high value source regions vs peripheral vs other regions doesn't necessarily imply colocation if all builders are for example in the high value cluster. We can have high value cluster instances where 4/5 instances are in europe, or where all sources are dispersed.}

% \parhead{Cluster coverage}
Finally, for a source cluster $G\subseteq\mathcal I$,  we define \emph{value-weighted cluster coverage} by
% To measure source-side transaction coverage, we report coverage separately for the high-value and peripheral source clusters.
% Using the source-level expected coverage $\bar f_I(\mathbf{s})$ from \Cref{def:coverage}, the cluster coverage of a source cluster $G\subseteq\mathcal I$ is
\[
    \mathrm{Coverage}(G;\mathbf{s})
    =
    \frac{
        \sum_{I\in G} \Lambda_I \bar f_I(\mathbf{s})
    }{
        \sum_{I\in G} \Lambda_I
    }.
\]
We report this metric for $G=\mathcal I_H$ and $G=\mathcal I_P$. Higher coverage means that transactions from that cluster are more likely to reach at least one builder before the deadline.
%while lower coverage indicates that transactions from that cluster are less likely to reach any builder before the deadline.

% \ignore{
% \noindent\textbf{Builder distances.} 
% Finally, we report the mean pairwise geographic distance among builders.
% Let $d_{\mathrm{geo}}:\mathcal R\times\mathcal R\to\mathbb R_{\ge 0}$ denote the great-circle distance between two regions, computed from the latitude and longitude of the GCP regions. For a final builder profile $\mathbf{s}$, we define
% \[
%     \mathrm{Dist}_{\mathrm{pair}}(\mathbf{s})
%     =
%     \frac{2}{K(K-1)}
%     \sum_{1\le b < b' \le K}
%     d_{\mathrm{geo}}(\mathbf{s}_b,\mathbf{s}_{b'}).
% \]
% This metric measures the average physical distance between builders and complements geographic HHI: larger values indicate that builders are more geographically spread out, while a value of zero means that all builders are located in the same region.
% \sen{perhaps remove this one}\boz{uyes}
% }
\section{Simulation Results}
\label{sec:results}
% How do protocol design levers affect equilibrium welfare and concentration? What's the impact of value asymmetry across regions?
We now study how the equilibria behave across concrete source-value and latency environments. The goal is to understand how self-interested builder location choice affects welfare, transaction coverage, geographic centralization, and builder utility concentration in equilibrium, and how these outcomes compare with those in the welfare-optimal planner benchmark. 
The experiments vary the two main forces in the model.
Source-value asymmetry changes builders' private incentives by making some sources more attractive than others. The slot duration controls how strongly latency affects whether transactions arrive before the deadline.

\begin{description}
    \item[Experiment 1: source-value asymmetry.] We isolate the effect of source-value asymmetry by sweeping the high-to-low source value ratio at a fixed slot duration and builder count.
    \item[Experiment 2: latency deadline effect.]  We isolate the latency deadline effect by sweeping the slot duration while holding value asymmetry fixed.
    \item[Experiment 3: interaction effect.] We study the interaction between source-value asymmetry and latency sensitivity by jointly varying the value ratio and the slot duration.
    \item[Experiment 4: builder participation.] We vary the number of builders to test whether additional builders expand coverage or mainly duplicate coverage around attractive sources.
\end{description}

The main figures in these experiments report aggregate metrics across randomized source layouts and initializations. These summaries are useful for comparing welfare, coverage, and concentration, but they do not show the underlying builder placements in individual simulation runs. To make the placement mechanism more concrete, we visualize a representative builder placement in equilibrium and the planner benchmark for each experiment in \fullref{app:placement-view}.

\subsection{Experiment 1: source-value asymmetry effect}
% Question/setup.
% Results.
% Interpretation.
This experiment isolates the effect of source-value asymmetry. We test whether increasing the value gap between high-value and peripheral sources causes builders to concentrate near high-value sources, creating redundant coverage and reducing peripheral-source coverage.

\parhead{Setup} We fix $K=5$ builders and $\Delta=50$ms slot duration.\footnote{We use $K=5$ as a baseline, motivated by the fact that the current Ethereum builder market is concentrated among a few builders~\cite{oz2024wins,yang2025decentralization,relayscan}. $50$ms slot duration is motivated by Solana Constellation's configuration, in which proposers submit transactions and the leader assembles a batch in each cycle~\cite{kniep_solana_constellation_2026}.}
Each simulation instance contains five high-value sources and five peripheral sources, sampled from their respective region pools.
We vary the high-to-low source value ratio $\gamma$ from $1$ to $20$, while keeping the total expected emitted value fixed. Thus, increasing the ratio reallocates value toward the high-value sources without changing the total value generated.
% All sampling, initialization, ABR termination, and metric-aggregation procedures follow the standard simulation workflow.

% \sen{how about the following figures? I only keep welfare ratio, two HHIs and the cluster coverage}\luis{I like it}

\begin{figure}
    \centering
    \includegraphics[width=\linewidth]{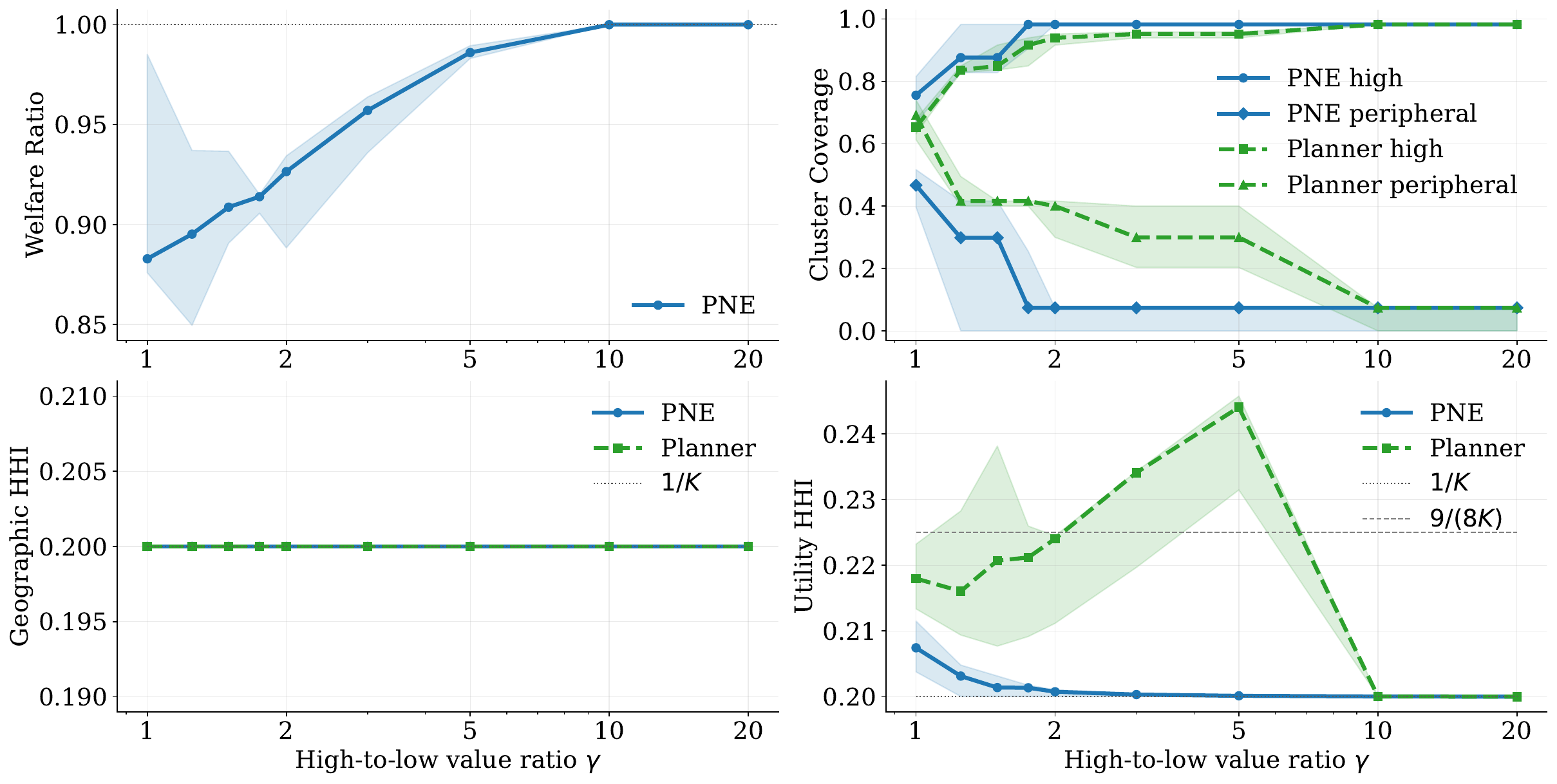}
    \caption{\textit{Effect of source-value asymmetry.} Lines show medians across runs; shaded bands show inter-quartile ranges. As the high-to-low source value ratio increases, equilibrium welfare approaches the optimal welfare while geographic \gls{hhi} remains exactly $1/K$. The main effect is coverage reallocation: at intermediate value ratios, coverage already concentrates on high-value sources in equilibrium, whereas the planner maintains peripheral source coverage and induces higher utility concentration.}
    \label{fig:exp1}
\end{figure}

\parhead{Results} In the top-left panel of \Cref{fig:exp1}, the welfare in equilibrium approaches that in the planner benchmark as the high-to-low value ratio increases. The welfare ratio rises toward one, indicating that builder placement in equilibrium recovers most of the optimum welfare achieved by the planner when high-value sources dominate the total value.

This improvement is not driven by region-level co-location. The bottom-left panel shows that geographic \gls{hhi} remains essentially at the uniform baseline $1/K$ throughout the value-ratio sweep. In this experiment, there are exact $K=5$ builders and five high-value sources, 
so builders can cover the high-value cluster in equilibrium by spreading across those regions rather than co-locating.
%and five peripheral information sources，
%Since the number of high-value sources matches the number of builders, builders can respond to larger high-value rewards by spreading across distinct high-value sources rather than co-locating in the same region.
The main effect of source-value asymmetry is therefore the reallocation of coverage. In the top-right panel, as the high-to-low value ratio increases, builders become increasingly incentivized to cover high-value sources in equilibrium: high-value cluster coverage quickly approaches one, while peripheral cluster coverage drops sharply. In the planner benchmark, the planner keeps peripheral coverage higher at intermediate value ratios, but this comes with a higher builder utility HHI (bottom-right panel), indicating that welfare-improving coverage can require allocating builders to lower-payoff regions.

Overall, in this balanced-capacity setting, source-value asymmetry reallocates coverage toward high-value sources without inducing region-level co-location: with as many high-value source regions as builders, builders can spread across distinct high-value regions. This should not be interpreted as saying that value asymmetry cannot create geographic concentration. If high-value sources are themselves geographically clustered, the same incentive to cover high-value sources may instead lead builders to concentrate in a small number of regions.

% \begin{figure}
% \centering    \includegraphics[width=\linewidth]{figures/builder_placement.png}
%     \caption{example fig}
% \end{figure}

\subsection{Experiment 2: slot-duration effect}
% Same structure.
Next, we fix the source-value distribution and vary the slot duration, isolating the other main force in the model: deadline-induced latency sensitivity. 
This experiment tests how the block-construction round duration changes latency-sensitive builder placement. When the duration is short, propagation latency has a stronger effect on whether transactions reach builders before the deadline. When the duration is long, more source-builder pairs can receive transactions in time, and geographic placement should matter less for coverage.
% \sen{do we want to clarify the transaction volume here does not affect the outcome?}

\parhead{Setup}
We fix $K=5$ builders and high-to-low value ratio $\gamma = 10$.
% Since the high-value and peripheral clusters contain the same number of sources, this corresponds to a high-value share of approximately $\alpha=\frac{10}{10+1}\approx 0.909$.
Each instance contains five high-value sources and five peripheral sources.
We vary slot durations from $10$ms to $12$s.  The lower end is tighter than the proposed $50$ms Solana Constellation cycle \cite{kniep_solana_constellation_2026}, while the upper end matches the $12$s slot duration used by Ethereum Proof-of-Stake~\cite{buterin2020combining}. When sweeping slot durations $\Delta$, we hold each source's expected emitted value $\Lambda_I$ fixed, so the experiment isolates the effect of the deadline on transaction reception rather than changing the amount of value generated per round.
% All sampling, initialization, ABR termination, and metric-aggregation procedures follow the standard simulation workflow.

\parhead{Results} \Cref{fig:exp2} shows a non-monotone effect of slot duration. In the top-left panel, when slots are very short, the welfare ratio is close to one, but this does not indicate broad coverage. Rather, the deadline is so tight that builders can effectively capture value only from nearby sources. Since high-value sources dominate the value distribution in this experiment, both the planner and the equilibrium incentive place builders near high-value source regions, leaving peripheral sources mostly uncovered.

% the duration is so tight that peripheral sources are largely unreachable for both the ABR-selected PNE and the planner benchmark.
% In this regime, the planner has little additional geographic flexibility to use.

\begin{figure}
    \centering
    \includegraphics[width=\linewidth]{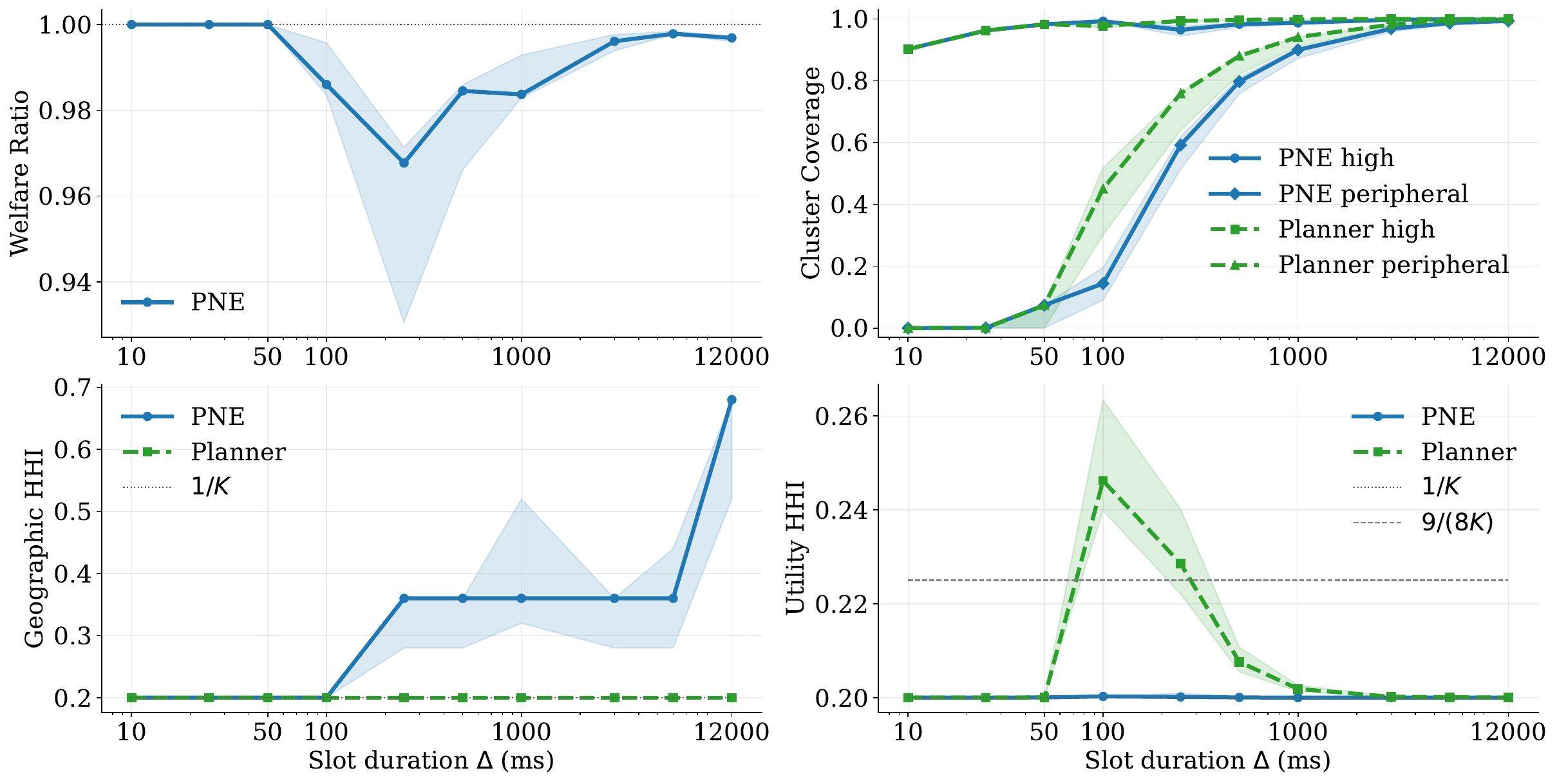}
    \caption{\textit{Effect of slot duration.} Larger slot durations relax the latency deadline and increase source-cluster coverage. The welfare gap is largest at intermediate durations, where peripheral coverage is feasible in the planner benchmark but not reliably supplied in equilibrium. At long durations, coverage is nearly saturated, so geographic concentration has limited welfare cost even when the builder placement in equilibrium remains more concentrated than in the planner benchmark.}
    \label{fig:exp2}
\end{figure}

% As the slot duration enters the intermediate range, peripheral transactions become reachable only from some builder locations. The planner can therefore increase welfare by allocating the fixed set of five builders across regions that are closer to different source clusters, rather than placing multiple builders near the same high-value sources. This is visible in geographic\gls{hhi} panel: ABR geographic concentration rises relative to the planner precisely in the range where the welfare ratio falls.

% The cluster coverage panel shows the mechanism behind this welfare loss.
% High-value cluster coverage is already near saturation, while peripheral cluster coverage under ABR lags behind the planner.
% Thus, the welfare loss comes from redundant coverage around already attractive sources rather than from a failure to cover the high-value cluster.

At intermediate slot durations, peripheral sources become reachable for some regions, but coverage still depends strongly on builder placement. This is where the welfare ratio falls: the planner spreads builders across regions to improve coverage of different source clusters, while in equilibrium, builders are concentrated around regions with better access to high-value sources. The top-right cluster-coverage panel shows the resulting mechanism: high-value coverage is already near saturation, whereas peripheral coverage in equilibrium lags behind that in the planner benchmark. Thus, the welfare gap comes primarily from under-coverage of peripheral sources in equilibrium, induced by redundant coverage around already attractive high-value sources. The utility HHI panel indicates that this welfare gap is not primarily an effect of payoff inequality: builder utilities remain relatively balanced even when peripheral coverage is under-provided.

This transition is best interpreted relative to the latency scale of the underlying region graph (see \fullref{app:gcp-regions}). The largest welfare losses arise in the intermediate regime, where $\Delta$ is comparable to cross-region latencies: peripheral coverage is feasible, but only for appropriately placed builders. When $\Delta$ is much smaller than typical cross-region propagation delays, builders can capture value only from nearby sources, so builders have limited ability to cover peripheral sources in both the equilibrium and the planner benchmark. When $\Delta$ is much larger than these latencies, most source-builder pairs can communicate before the deadline, so transaction coverage becomes less sensitive to builder geography. In this long-slot regime, geographic centralization may persist, but it no longer produces a large welfare loss in this model because coverage is close to saturated.

% The welfare ratio recovers when slot duration becomes longer.
% At these deadlines, most source-builder pairs have enough time for transactions to arrive before the deadline. As a result, even if the equilibrium profile remains geographically centralized, it still covers nearly all the emitted value. In other words, longer slots make transaction coverage less sensitive to builder geography: centralization may persist, but it no longer produces a large welfare loss. The bottom-left utility \gls{hhi} panel tells a similar story. Builder payoffs remain balanced in equilibrium, so the welfare gap in the intermediate regime is not mainly about payoff inequality. It is about source coverage: peripheral sources are under-covered in equilibrium than in the planner benchmark when the latency constraints are neither extremely tight nor fully relaxed. 
We repeat this experiment with a lower value ratio of $\gamma = 2$, where peripheral sources contribute to a larger share of the total value. The qualitative patterns are similar, but the welfare gap is larger at intermediate slot durations, where peripheral sources are valuable enough to matter and reachable enough for the planner to cover. We provide the figure and discussion in \fullref{app:exp2_low_gamma}.

% Thus, the relevant design parameter is not slot duration in isolation, but slot duration relative to source-to-builder propagation latencies. Protocols with very short rounds can make transaction coverage highly local, while protocols with long rounds reduce the welfare cost of geographic concentration. The fragile regime is the intermediate one, where cross-region coverage is possible but still sensitive to builder placement.

% propagation delay is no longer the binding constraint for most source-builder pairs, so builders can receive nearly all emitted value before the deadline despite being geographically centralized in equilibrium.\boz{this is a critical discussion item; longer durations make participation/coverage less-geography dependent. eg solana vs ethereum}\fei{should be cautious under the equal split sharing rule}\boz{ofc under this modeling with this rule, all these results apply}. \stef{I agree this is an important point as slot duration is a central consideration currently (and for some time) for all blockchains. As Figure 2 shows, however, this may lead to geographic concentration introducing the usual risk associated with that.} Geographic \gls{hhi} may remain high, but its welfare cost becomes small because coverage is close to saturated. The utility \gls{hhi} panel reinforces this interpretation: payoff inequality among builders remains relatively low, so the main inefficiency is not unequal builder utilities, but the spatial allocation of coverage when latency constraints are partially, but not fully, relaxed.

\subsection{Experiment 3: value-asymmetry--slot-duration interaction effect}
Having isolated source-value asymmetry and slot duration separately, we now study their interaction. The goal is to identify which combinations produce inefficient welfare outcomes, weak transaction coverage, or concentrated builder placement in equilibrium.

\parhead{Setup}
Each simulation instance contains five builders, five high-value sources, and five peripheral sources, sampled from their respective region pools.
We sweep the high-to-low source value ratio $\gamma$ from $1$ to $20$ and the slot duration from $10$ms to $12$s.
% For each parameter pair, we follow the standard simulation workflow with five source layouts and three random initial builder profiles per layout.

\parhead{Results}
\Cref{fig:exp3} shows that source-value asymmetry and slot duration jointly determine how close the equilibrium outcome is to the welfare-optimal outcome, and reinforces our prior results. The welfare ratio is primarily non-monotone in the slot duration. At both extremes, welfare in equilibrium is close to that in the planner benchmark. The largest welfare loss occurs at intermediate slot durations, where peripheral sources are reachable from certain regions but coverage remains sensitive to builder placement. The value ratio affects how costly this peripheral under-coverage is. When the high-to-low value ratio is small, and the peripheral sources meaningfully contribute to the total emitted value, the welfare loss becomes larger in equilibrium. As the value ratio increases, high-value sources dominate the total emitted value, so the welfare penalty from leaving peripheral sources under-covered becomes smaller.

The coverage-gap surfaces identify the mechanism. A positive gap means the corresponding source cluster is better covered in the planner benchmark than in equilibrium. When the value ratio $\gamma=1$, the two source clusters have equal expected value, so the labels ``high-value'' and ``peripheral'' distinguish geography rather than value. In this case, planner advantage in one cluster and equilibrium advantage in the other need not create a large welfare gap.

As the value ratio increases, the high-value cluster receives a larger share of total value, and builders in equilibrium have stronger incentives to prioritize regions with access to high-value sources. At intermediate slot durations, peripheral sources are still reachable from some regions, so the planner can recover additional welfare by allocating builders to improve peripheral coverage. By contrast, peripheral coverage is under-provided in equilibrium, and builder placement remains tilted toward high-value sources, contributing to a negative coverage gap in high-value source coverage, a positive gap in peripheral source coverage, and a corresponding welfare penalty. 

The \gls{hhi} surfaces are consistent with this mechanism. Utility concentration becomes higher in the planner benchmark because welfare-improving allocations yield lower individual payoffs for builders placed away from high-value sources. Conversely, the utility \gls{hhi} in equilibrium remains balanced and close to the egalitarian baseline in much of the grid, but it comes with a more geographically centralized builder placement around the high-value sources.

% Instead, the largest welfare losses appear in an intermediate region: slot durations are long enough that peripheral sources can be reached from some builder locations, but short enough that geographic placement still determines which sources are covered. At the same time, peripheral sources remain valuable enough that losing their coverage has a meaningful welfare cost.

\begin{figure}[t]
\centering    \includegraphics[width=\linewidth]{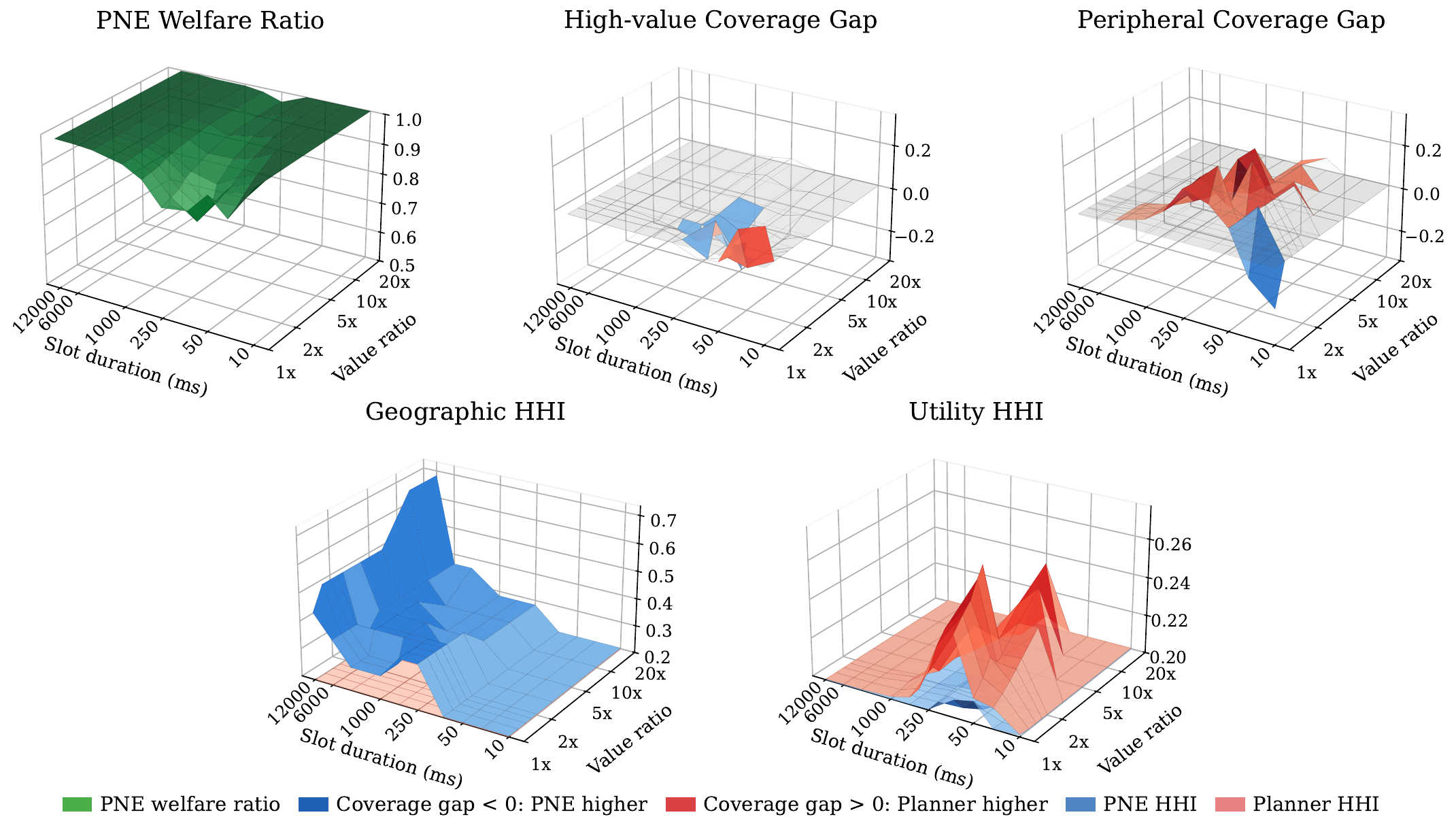}
\caption{\emph{Joint effect of source-value asymmetry and slot duration.}
Surfaces show medians across source layouts and initializations for fixed
\(K=5\). Darker colors indicate higher absolute values. Coverage gaps are defined as planner coverage minus equilibrium coverage, so positive gaps indicate that the planner covers more of the corresponding source cluster. The HHI panels compare geographic and utility concentration in equilibrium and in the planner benchmark. Welfare losses are largest at intermediate slot durations and lower-to-moderate value ratios.}
    \label{fig:exp3}
\end{figure}

\subsection{Experiment 4: builder participation effect}

Previous experiments isolate the effects of source-value asymmetry and slot duration, as well as their interaction.
We now vary the number of builders to test whether larger builder sets expand source coverage or mainly duplicate coverage around already attractive sources. Adding builders increases the system's potential coverage capacity, but when source values are asymmetric, additional builders may still prefer regions with strong access to high-value sources.

\parhead{Setup} We vary the number of builders from $K=3$ to $K=12$. We fix the high-to-low value ratio at $\gamma=10$ and the slot duration at $\Delta=50$ms. 
For this builder-count sweep, we use the greedy welfare-maximization routine as the planner benchmark, since the cost of exact search grows rapidly with the builder count $K$.
% Each simulation instance contains six high-value sources and six peripheral sources, sampled from their respective region pools.
% \sen{do we really need 6 sources?}
% \luis{We can keep it consistent with the other experiments and have 5 in each cluster. I ran with both 10 and 12 total info sources just to see if the same patterns would arise}
% For each builder-count setting, we follow the standard simulation workflow with five source layouts and three random initial builder profiles per layout.

\parhead{Results}
\Cref{fig:exp4} shows that increasing the number of builders does not automatically yield broader source coverage in equilibrium. In the top-left panel, the welfare ratio is non-monotone in $K$. When the number of builders is small, additional builders mainly improve high-value source coverage, so the welfare in equilibrium quickly moves closer to the optimal welfare. Once high-value coverage is nearly saturated, further welfare gains require better peripheral-source coverage. 

The top-right cluster-coverage panel shows that the planner uses additional builders to expand peripheral-source coverage. In equilibrium, however, they remain attracted to regions with strong access to high-value sources, even after those sources are nearly saturated. As a result, larger builder sets can generate redundant high-value coverage rather than broad source coverage, explaining why the welfare ratio first improves and then declines.

\begin{figure}
    \centering
    \includegraphics[width=\linewidth]{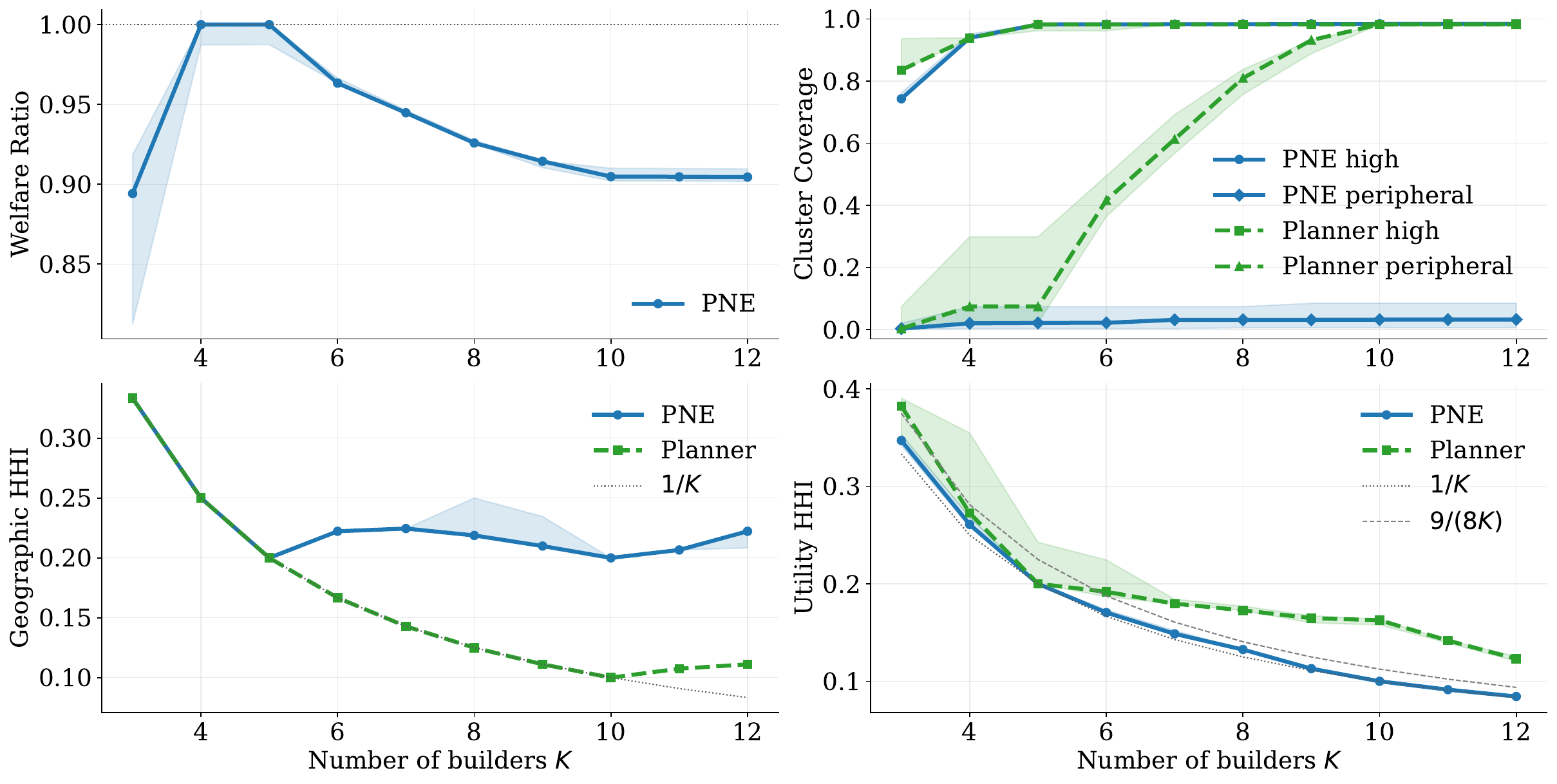}
    \caption{\textit{Effect of builder participation.}  As the number of builders $K$ increases, high-value source coverage is improved in equilibrium, while peripheral coverage remains low, contributing to the welfare ratio first increasing then declining. Geographic \gls{hhi} shows that the builder locations remain more spatially concentrated in equilibrium than the planner benchmark for larger $K$, while utility \gls{hhi} remains comparatively close to the egalitarian baseline.}
    \label{fig:exp4}
\end{figure}

The \gls{hhi} panels reveal a tradeoff between geographic decentralization and utility equality. As the number of builders increases, geographic \gls{hhi} moves closer to the minimum $1/K$ in the planner benchmark, because additional builders expand coverage toward peripheral sources. Builder placement in equilibrium remains more geographically centralized: additional builders continue to locate in regions with stronger access to the high-value cluster, even after high-value coverage is nearly saturated.

The pattern is reversed for utility concentration. Utility \gls{hhi} is closer to the egalitarian baseline in equilibrium, while it is higher in the planner benchmark. Improved peripheral source coverage in the planner benchmark achieves higher total welfare, but yields lower individual utility for those builders located in the peripheral regions. By contrast, in equilibrium, builders are incentivized to locate in high-value regions, where equal sharing can make payoffs more balanced. Thus, geographic decentralization and builder-utility equality need not coincide: improving source coverage and total welfare may require distributing builders into lower-payoff regions, while selfish behavior can achieve a more utility-balanced but geographically more centralized outcome in equilibrium. 

We repeat this experiment with a lower value ratio of $\gamma = 2$, where peripheral sources contribute to a larger share of the total value. The qualitative patterns are similar, but the welfare gap is more pronounced. We provide the figure and discussion in \fullref{app:exp2_low_gamma}.

% The \gls{hhi} panels support the same interpretation.
% Geographic \gls{hhi} falls with $K$ for both ABR and the planner, partly because the benchmark\luis{Do we mean floor/minimum achievable HHI?} $1/K$ also decreases.
% However, the planner stays closer to this benchmark, whereas ABR remains more geographically concentrated once $K$ exceeds the number of high-value locations.
% Utility \gls{hhi} is also higher under ABR over much of the larger-$K$ range, indicating less even payoff distribution across builders.

% Thus, additional participation does not by itself eliminate geographic concentration or payoff inequality when private incentives favor the high-value cluster.

% \boz{I think it's the other way around? the planner hhi is higher because it would place builders to "bad" regions, which won't happen in reality. Similar to results in exp1. I think in general, as utility asymmetry grows, geo hhi and utility hhi oppose; planner puts builders to periphery, reducing geo hhi but increasing uility hhi, whereas ABR outcome is the other way around. This effect is amplified under shorter slots probably.}
% \boz{revise!}

\section{Discussion}\label{sec:discussion}
We now discuss the implications of our results for decentralized block-building design. The central message is that decentralization is not only a question of how many builders participate, but also of where self-interested builders choose to locate. In our model, the \emph{price of decentralization} \revised{is the welfare loss that arises when builders choose regions according to individual incentives rather than being coordinated to maximize system-wide transaction coverage.}{measures the welfare cost of decentralized, self-interested location decisions relative to a centrally coordinated placement that maximizes transaction coverage. Decentralized decisions may themselves produce a geographically concentrated outcome.} We then discuss directions for future work and limitations.

\subsection{Implications for decentralized block-building design}
Our theory and simulations point to a structural feature of decentralized block-building or \gls{mcp} designs. 
% adding multiple builders is not sufficient by itself to guarantee broad transaction coverage or welfare-efficient placement. 
Multiple builders reduce the monopoly power of a single proposer, but if their location choices are uncoordinated, equilibrium placement can still duplicate coverage around privately attractive high-value sources, weakening the system's censorship-resistance and fair-access guarantees.

The simulations identify where this price of decentralization is most visible. Welfare losses are largest in intermediate regimes where peripheral sources are reachable and still contribute meaningful value, but decentralized equilibrium incentives still favor regions with strong access to high-value sources. Thus, inefficiency is not driven simply by source-value asymmetry or by latency sensitivity alone, but by their interaction: peripheral coverage must be feasible enough to matter, while still requiring builders to locate away from the most privately attractive regions.

% From a protocol-design perspective, parameters that affect latency sensitivity can shape the value of geographic placement.
From a protocol-design perspective, the relevant timing parameter is not slot duration in isolation, but slot duration relative to transaction propagation latencies.
Longer slot durations or more relaxed (sub-)block construction deadlines can make transaction coverage less geography-dependent in our model, but geographic centralization may persist when its welfare penalty becomes small. Conversely, shorter deadlines make source--builder proximity more important, while transaction coverage is restricted to be local, limiting coverage of distant sources. The fragile regime lies between these extremes, where cross-region coverage is possible but still placement-sensitive.

Increasing builder participation expands potential coverage capacity. However, in the finite regimes our simulations examine (up to 12 builders), additional builders may not cover peripheral sources immediately. When individual incentives favor high-value source access, extra builders can duplicate coverage around already attractive regions. The results also show that geographic and utility concentration capture different objectives. Planner allocations may be geographically more dispersed but utility-wise more unequal, because the planner places builders in lower-payoff peripheral regions to improve total coverage. Equilibrium profiles may instead be more geographically concentrated while keeping builder utilities relatively balanced under equal sharing.

Since source-value asymmetry is largely an environmental primitive, while timing, builder participation, and reward-sharing rules are at least partly protocol design choices, these parameters are natural levers for shaping transaction coverage incentives.
% protocol design can more directly affect the price of decentralization through timing, participation rules, and reward-sharing rules. 
Overall, protocol designers should distinguish between maximizing welfare and coverage, improving geographic decentralization, and equalizing builder rewards, since these objectives need not move together.

\subsection{Future work}
\parhead{Location-market modeling}
Our model treats decentralized block building as a stochastic coverage game. Transaction sources are exogenous: they emit transactions, builders choose regions, and latency determines whether each transaction reaches at least one builder before the deadline. This formulation is useful for studying source-side coverage and the welfare loss from redundant coverage of the same sources.

A richer model would endogenize transaction routing. In such a model, users or applications choose which builders to send transactions to, and the builder location affects which transactions a builder can serve with sufficiently low latency. Builder rewards would then depend not only on which sources are reachable, but also on how users route order flow and how builders compete for that flow. This would make the model closer to a location-market game than a pure coverage game. 

Such a formulation may be important for protocols and applications in which latency-sensitive users actively choose among builders, such as the Solana ecosystem with Constellation~\cite{kniep_solana_constellation_2026}. High-value order-flow sources, such as exchanges or latency-sensitive application servers, may emerge as strategic hubs: co-locating builders may receive higher-quality order flow or offer stronger latency guarantees. We leave the analysis of endogenous transaction routing and builder competition for future work.

% Our model treats decentralized block building as a stochastic coverage game, emphasizing whether each transaction reaches at least one builder before the deadline. An alternative approach is to model decentralized block building as a location-based market game, in the spirit of facility-location games~\cite{vetta2002utilitygames,Tardos_Wexler_2007}. In such a model, users or applications route transactions to reachable builders, and builder location determines market reachability or service quality through latency. Competition among builders that can serve the same market then determines the inclusion fee or reward that can be charged or shared.

% This perspective is closer to transaction-routing designs such as Solana's Constellation~\cite{kniep_solana_constellation_2026}. High-value order-flow sources, such as exchanges or latency-sensitive application servers, may emerge as economically significant hubs, and builders colocated with or in close network proximity to these hubs may obtain access to higher-quality order flow or be able to offer tighter latency guarantees. The resulting equilibrium geography of builders can differ from the coverage-game formulation analyzed in this work. However, provided that the induced strategic interaction continues to satisfy the axioms of valid utility games~\cite{vetta2002utilitygames}, analogous worst-case bounds on efficiency loss may still be derived.

\parhead{Alternative reward-sharing rules}
Our analysis focuses on equal splitting among builders who receive and include a transaction. Future work could compare this rule with other reward-sharing mechanisms. Winner-takes-all or unique-inclusion rules may reduce duplicate inclusion, but can strengthen incentives for exclusivity, private routing, and reward concentration. Broader committee-level sharing \cite{pranav2025tfmmcp} may encourage builders to maximize aggregate source coverage, but can create free-riding incentives if rewards are paid to builders that do not contribute to coverage. 

Another natural extension is contribution-weighted or Tullock-style sharing.  The equal-split rule can be viewed as the uniform-score special case of proportional sharing. More general rules could assign contribution scores, for example, based on the source-builder propagation latency, and allocate value proportionally to the score. Such rules interpolate between egalitarian sharing and winner-takes-all allocation, and relate to  work on the blockspace distribution and contest-style mechanisms~\cite{mike2024blockspacedistribution,garimidi_tullock_2025,garimidi2026winnertakeallprocurementauctions}. Studying their effects on welfare, geographic centralization, and reward equality is an important direction for future work.

\subsection{Scope and limitations}
% \done{\stef{This subsection should be rewritten as as a rebuttal to these precise limitations: when stating a limitation, we should explain why this is not important for our how it affects our setting and results. For example: We assume no block-capacity. In reality, this implies that the location-selection game is no longer an exact potential. However (and this is the important part), this would only strengthen our negative/positive result XY  (e.g., low/high HHI, latency effects etc) derived under the potential assumption. Or another example: Empirical calibration would be better, however such data is currently unavailable and we have used the most reliable sources, \cite{dataalways2025geography}. The rough estimates derived from this and other sources, have informed our parameter ranges which cover any plausible practical scenario. Another example: While protocols may impose local deadlines, this significantly complicates the block-building process and generates intricate incentives and dynamics requiring a substantially different approach. Thus, while our model does not cover this case, this does not provide a limitation against the study of current mainstream protocols. Even so, these local deadlines would only amplify the effects XY that lead to our result Z. And so on.}}
We close by clarifying the scope of the model and simulations. The goal of the paper is not to provide a full protocol-level simulator, but to isolate how latency-sensitive transaction value shapes the region choices of self-interested builders in decentralized block-building.

First, our simulations use asynchronous better-response dynamics to select a pure Nash equilibrium. We use this dynamic because it is theoretically aligned with the exact-potential structure of the game under equal sharing and gives a stable equilibrium benchmark. Simultaneous updates or learning dynamics may cycle or induce non-stationary empirical distributions; while studying such out-of-equilibrium behavior is an interesting direction, it is outside the scope of this paper. We use ABR as an equilibrium-selection procedure, not as a literal description of how all builders relocate in practice. The underlying incentive tension remains the same: builders prefer regions with stronger access to certain high-value sources, while the system may benefit from broader source coverage. 

Our exact-potential argument guarantees finite convergence of strict asynchronous better-response dynamics, but it does not provide a polynomial bound on the length of improvement paths. Thus, the complexity of computing a pure Nash equilibrium in the builder region game is a natural direction for future work. For finite, explicitly represented instances with polynomial-time evaluable utilities, pure Nash equilibrium computation can be viewed as a local-search problem over the potential function, suggesting that a formal PLS analysis may be appropriate. In our simulations, however, ABR served only as an equilibrium-selection procedure, and all reported runs converged to a no-improving-deviation profile within 250 builder updates. Moreover, the smoothness guarantee in \Cref{thm:poa} extends beyond pure Nash equilibria to weaker equilibrium notions, so the factor-\(2\) efficiency bound is not tied only to the ABR selection procedure.

% This assumes that builders update locations one at a time, which may not literally describe decentralized relocation decisions in practice. Simultaneous responses can create delayed or cyclic adjustments, especially when builders react to stale or perturbed information. However, this affects the path of adoption than the underlying incentives of builders to prefer regions with stronger access to high-value sources. Moreover, the welfare guarantee in \Cref{thm:poa} extends beyond pure Nash equilibria to weaker equilibrium notions, as noted in \Cref{rem:broader_eq}. Thus, converged equilibria should be interpreted as a stable benchmark for studying the location incentives, rather than as a literal description of every adjustment path.

Second, we abstract from many communication details. We model source-to-builder latency directly and use a common deadline for transaction reception. In practice, protocols may involve peer-to-peer propagation, intermediaries (e.g., relays), committee attestation, and leader aggregation. 
%local deadlines for reaching an aggregating leader or committee, and incentives to reach attesters or validators. 
These features would require a richer model of consensus-layer and network communication. However, they do not remove the central mechanism studied here: location affects timely access to payoff-relevant information. Additional latency-sensitive targets may, in fact, create more location incentives~\cite{yang_geographical_2025}, reinforcing the same kind of concentration pressure.

% further incentives to locate near high-value sources or other payoff-critical entities, similar to the concentration incentives observed in Ethereum block-building architectures~\cite{yang_geographical_2025}.

Third, our simulations use stylized source locations and value asymmetries rather than direct empirical estimates of order-flow geography. Such calibration is difficult because high-value transaction sources and their value distribution are not directly observable.
%it requires identifying high-value transaction sources and estimating their value distribution. To mitigate this, 
We therefore use empirical observations on geographic latency competition in Ethereum block building~\cite{dataalways2025geography} to instantiate the region graph and sweep broad ranges of value asymmetry, slot duration, and builder participation. The simulation results should be read as comparative statics over plausible environments, rather than as a point estimate for a particular network.

Finally, we omit several strategic features that are important for full protocol design. We impose no block-capacity constraint, fix the builder set, and rule out Sybil behavior and collusion. Capacity constraints would couple transaction selection decisions and may strengthen the welfare-loss channel we identify, because redundant inclusion of already-covered high-value transactions could crowd out otherwise uncovered transactions. Entry and rotation would affect which builders are active in a given round and could change equilibrium selection. Sybil behavior and collusion would require a richer model of coalitional incentives, where coordinated builders may improve welfare by avoiding redundant coverage, but they could also undermine the censorship-resistance benefits of decentralized block building by behaving as a single economic entity. These extensions are therefore important for full protocol design, but complementary to the location-choice mechanism studied here, which focuses on the inefficiency of uncoordinated self-interested placement.

\section{Conclusion}\label{sec:conclusion}
We studied location choice in decentralized block building by modeling it as a stochastic coverage game. Our results show that replacing a single proposer with multiple builders is not sufficient, by itself, for broad transaction coverage or welfare-efficient placement. Under equal sharing, the builder region game admits a pure Nash equilibrium, has an asymptotically tight factor-$2$ \gls{poa} bound, and satisfies tight bounds on builder utility concentration. Our simulations complement these worst-case guarantees by showing that welfare losses are most pronounced in intermediate regimes where peripheral sources are reachable and valuable, but self-interested builders still prefer regions with strong access to high-value sources. Overall, the paper highlights that the benefits of decentralized block building depend not only on the number of builders, but also on where builders locate and how protocol rules shape their incentives. Timing parameters, participation rules, and reward-sharing mechanisms can therefore affect the price of decentralization, especially in networks with asymmetric source values and latency-sensitive transaction propagation.

%%
%% Bibliography
%%

%% Please use bibtex, 

\bibliography{reference}

\iffullversion
    % \crefalias{section}{appendix}
    % \crefalias{subsection}{appendix}
    % \crefalias{subsubsection}{appendix} 
    \appendix
    \section{Omitted Proofs}\label{app:proofs}

\begin{proof}[Proof of \Cref{prop:potential}]
Fix a builder $b$, a profile $\mathbf{s}$, and a deviation $s_b'\in\mathcal{R}$. We compare the expected change in potential with the expected change in builder $b$'s reward.

Consider a realized transaction $j$ emitted by source $I$ at time $t_j$, with value $V_j$. Conditional on the realized transaction set, values, emission times, and the reception outcomes of all builders other than $b$, define
\[
n_j^{-b}
:=
\left|\{c\in\mathcal{B}\setminus\{b\}: j\in\mathcal{T}_c\}\right|.
\]

A unilateral deviation by builder $b$ does not affect which transactions are emitted, their values or emission times, or the reception outcomes of builders $c\neq b$. Therefore, $n_j^{-b}$ is the same under $\mathbf{s}$ and under $(s_b',\mathbf{s}_{-b})$.

If builder $b$ chooses region $s_b$, then it receives transaction $j$ with probability
\[
q_j(s_b) := q_{r(I)}(s_b,t_j).
\]
Conditional on $n_j^{-b}$, builder $b$'s expected reward from transaction $j$ is
\[
V_j \cdot \frac{q_j(s_b)}{n_j^{-b}+1}.
\]
Similarly, if builder $b$ deviates to $s_b'$, its conditional expected reward from transaction $j$ is
\[
V_j \frac{q_j(s_b')}{n_j^{-b}+1}.
\]
Thus, the conditional expected change in builder $b$'s reward from transaction $j$ is
\[
V_j
\frac{q_j(s_b')-q_j(s_b)}{n_j^{-b}+1}.
\]

Now consider the contribution of transaction $j$ to the potential. If builder $b$ chooses $s_b$, then the number of builders receiving $j$ is $n_j^{-b}$ with probability $1-q_j(s_b)$ and $n_j^{-b}+1$ with probability $q_j(s_b)$. Hence the conditional expected potential contribution of transaction $j$ is
\[
V_j\left[
(1-q_j(s_b))H_{n_j^{-b}}
+
q_j(s_b)H_{n_j^{-b}+1}
\right].
\]
After the deviation to $s_b'$, the corresponding contribution is
\[
V_j\left[
(1-q_j(s_b'))H_{n_j^{-b}}
+
q_j(s_b')H_{n_j^{-b}+1}
\right].
\]

Therefore, the conditional expected change in potential from transaction $j$ is
\begin{align*}
&V_j\left[
(1-q_j(s_b'))H_{n_j^{-b}}
+
q_j(s_b')H_{n_j^{-b}+1}
\right]
-
V_j\left[
(1-q_j(s_b))H_{n_j^{-b}}
+
q_j(s_b)H_{n_j^{-b}+1}
\right] \\
&=
V_j\bigl(q_j(s_b')-q_j(s_b)\bigr)
\bigl(H_{n_j^{-b}+1}-H_{n_j^{-b}}\bigr) \\
&=
V_j
\frac{q_j(s_b')-q_j(s_b)}{n_j^{-b}+1},
\end{align*}
where the last equality uses $H_{m+1}-H_m=1/(m+1)$.

This equals the conditional expected change in builder $b$'s reward from transaction $j$. Summing over all realized transactions and taking expectations over transaction generation and the reception outcomes of the other builders gives
\[
\Phi(s_b',\mathbf{s}_{-b})-\Phi(\mathbf{s})
=
u_b(s_b',\mathbf{s}_{-b})-u_b(\mathbf{s}).
\]
Thus, $\Phi$ is an exact potential. Since the action set $\mathcal{R}$ is finite for every builder, the game is a finite exact potential game.

\end{proof}

\begin{proof}[Proof of \Cref{lem:welfare_coverage}]
Fix a profile $\mathbf{s}$ and a source $I$. For a transaction emitted by source $I$ at time $t$, the probability that at least one builder receives it before the deadline is
\[
f_I(t;\mathbf{s})
=
1-\prod_{b\in\mathcal{B}}
\bigl(1-q_{r(I)}(s_b,t)\bigr),
\]
where the product form follows from conditional independence of receptions across builders.

Therefore, the expected value from source $I$ that is covered by at least one builder is
\[
\mathbb{E}\!\left[
\sum_{j\in\mathcal{J}_I}
V_j \mathbf{1}\{n_j\ge 1\}
\right]
=
\mathbb{E}\!\left[
\sum_{j\in\mathcal{J}_I}
V_j f_I(t_j;\mathbf{s})
\right].
\]
Using the expected value-flow assumption with
\[
h(t)=f_I(t;\mathbf{s}),
\]
we obtain
\[
\mathbb{E}\!\left[
\sum_{j\in\mathcal{J}_I}
V_j f_I(t_j;\mathbf{s})
\right]
=
\Lambda_I
\int_0^\Delta
\rho_I(t) f_I(t;\mathbf{s})\,dt
=
\Lambda_I \bar f_I(\mathbf{s}).
\]
Summing over all sources $I\in\mathcal{I}$ gives the result.
\end{proof}

\begin{proof}[Proof of \Cref{lem:aggregate_utility}]
For any realized transaction $j$, the total reward distributed across builders is
\[
\sum_{b\in\mathcal{B}}
V_j \frac{\mathbf{1}\{j\in\mathcal{T}_b\}}{n_j}
=
V_j\mathbf{1}\{n_j\ge 1\},
\]
with the convention that the summand is zero when $n_j=0$. Summing over all transactions and sources and taking expectations gives
\[
\sum_{b\in\mathcal{B}} u_b(\mathbf{s})
=
\mathbb{E}\!\left[
\sum_{I\in\mathcal{I}}\sum_{j\in\mathcal{J}_I}
V_j\mathbf{1}\{n_j\ge 1\}
\right]
=
W(\mathbf{s}).
\]
\end{proof}

\begin{proof}[Proof of \Cref{lem:smoothness}]
The proof proceeds in three steps.

\paragraph*{Step 1: Lower bound individual payoff by exclusive coverage.}
Consider builder $b$ deviating to $s_b^*$ while all other builders remain at $\mathbf{s}_{-b}$. Under the deviation profile $(s_b^*,\mathbf{s}_{-b})$, builder $b$ receives the full value of any transaction that no other builder receives. Hence, by retaining only these exclusively covered transactions, we obtain the lower bound
\begin{equation}\label{eq:exclusive_lb}    
u_b(s_b^*,\mathbf{s}_{-b})
\geq
\sum_{I \in \mathcal{I}} \Lambda_I \int_0^\Delta \rho_I(t)\,
q_{r(I)}(s_b^*, t)
\prod_{b' \neq b}\bigl(1 - q_{r(I)}(s_{b'}, t)\bigr)\, dt.
\end{equation}

\paragraph*{Step 2: A pointwise coverage inequality.}
Summing~\eqref{eq:exclusive_lb} over all $b\in\mathcal{B}$ gives a cross-profile coverage term: builder $b$ is evaluated at its optimal-region deviation $s_b^*$, while all other builders remain at their current regions $\mathbf{s}_{-b}$. It therefore suffices to show that, for each source $I$ and each $t \in [0,\Delta]$,
\begin{equation}\label{eq:step2_target}
\sum_{b\in\mathcal{B}} q_{r(I)}(s_b^*, t)
\prod_{b' \neq b}\bigl(1 - q_{r(I)}(s_{b'}, t)\bigr)
\;\geq\;
f_I(t;\, \mathbf{s}^*) - f_I(t;\, \mathbf{s}),
\end{equation}
where $f_I(t;\, \mathbf{s})$ is the coverage probability at time $t$ from \Cref{def:coverage}.

Since $1 - q_{r(I)}(s_b,t) \leq 1$, dropping the factor for builder $b$ from the product can only increase it:
\begin{equation}
\prod_{b' \neq b}\bigl(1 - q_{r(I)}(s_{b'},t)\bigr)
\geq
\prod_{b' \in \mathcal{B}}\bigl(1 - q_{r(I)}(s_{b'},t)\bigr)
= 1 - f_I(t;\, \mathbf{s}).
\end{equation} 
Substituting:
\begin{equation}\label{eq:substitute}
\sum_{b\in\mathcal{B}} q_{r(I)}(s_b^*,t)
\prod_{b' \neq b}\bigl(1 - q_{r(I)}(s_{b'},t)\bigr)
\;\geq\;
\Bigl(\sum_{b\in\mathcal{B}} q_{r(I)}(s_b^*,t)\Bigr)
\cdot \bigl(1 - f_I(t;\, \mathbf{s})\bigr).
\end{equation}
Next, the coverage probability under $\mathbf{s}^*$ is the probability that at least one builder receives the transaction. By the union bound, this probability is at most the sum of individual reception probabilities:
\begin{equation}\label{eq:union_bound}
f_I(t;\mathbf{s}^*)
=
1 - \prod_{b\in\mathcal{B}}\bigl(1 - q_{r(I)}(s_b^*,t)\bigr)
\leq
\sum_{b\in\mathcal{B}} q_{r(I)}(s_b^*,t).
\end{equation}
Combining \eqref{eq:substitute} and \eqref{eq:union_bound}:
\begin{equation}\label{eq:sub_union_combine}
\sum_{b\in\mathcal{B}} q_{r(I)}(s_b^*,t)
\prod_{b' \neq b}\bigl(1- q_{r(I)}(s_{b'},t)\bigr)
\;\geq\;
f_I(t;\, \mathbf{s}^*)\cdot \bigl(1 - f_I(t;\, \mathbf{s})\bigr).
\end{equation}
Since $f_I(t;\, \mathbf{s}^*) \leq 1$,
\begin{equation}\label{eq:step2_final}
f_I(t;\, \mathbf{s}^*)\cdot \bigl(1 - f_I(t;\, \mathbf{s})\bigr)
= 
f_I(t;\, \mathbf{s}^*) - f_I(t;\, \mathbf{s}^*) \cdot f_I(t;\, \mathbf{s})
\;\geq\;
f_I(t;\, \mathbf{s}^*) - f_I(t;\, \mathbf{s}),  
\end{equation}
which establishes \eqref{eq:step2_target}.

\paragraph*{Step 3: Integrate over sources and emission times.}
Combining Steps 1 and 2:
\begin{align}
\sum_{b\in\mathcal{B}} u_b(s_b^*,\, \mathbf{s}_{-b})
&\geq
\sum_{I \in \mathcal{I}} \Lambda_I \int_0^\Delta \rho_I(t)\,
\bigl[f_I(t;\, \mathbf{s}^*) - f_I(t;\, \mathbf{s})\bigr]\, dt \\
&=
\sum_{I \in \mathcal{I}} \Lambda_I
\bar f_I(\mathbf{s}^*) - \sum_{I \in \mathcal{I}} \Lambda_I
\bar f_I(\mathbf{s}) \\
&=
W(\mathbf{s}^*) - W(\mathbf{s}).
\end{align}
where the last equality follows from~\Cref{lem:welfare_coverage}. This proves the $(1,1)$-smoothness inequality.
\end{proof}

\begin{proof}[Proof of \Cref{thm:poa}]
Let $\mathbf{s}\in\mathcal{E}$ be a pure Nash equilibrium, and let $\mathbf{s}^*$ be a welfare-maximizing profile. For every builder $b\in\mathcal{B}$, unilateral optimality implies
\[
u_b(\mathbf{s})
\geq
u_b(s_b^*,\mathbf{s}_{-b}),
\]
where $s_b^*$ denotes builder $b$'s region in $\mathbf{s}^*$.

Summing over all builders and using~\Cref{lem:aggregate_utility}, we obtain
\[
W(\mathbf{s})
=
\sum_{b\in\mathcal{B}}u_b(\mathbf{s})
\geq
\sum_{b\in\mathcal{B}}u_b(s_b^*,\mathbf{s}_{-b}).
\]
By \Cref{lem:smoothness} with the benchmark profile $\mathbf{s}^*$,
\[
\sum_{b\in\mathcal{B}}u_b(s_b^*,\mathbf{s}_{-b})
\geq
W(\mathbf{s}^*)-W(\mathbf{s}).
\]
Therefore,
\[
W(\mathbf{s})
\geq
W(\mathbf{s}^*)-W(\mathbf{s}),
\]
which implies
\[
W(\mathbf{s}) \geq \frac{1}{2}W(\mathbf{s}^*).
\]
Taking the minimum over all $\mathbf{s}\in\mathcal{E}$ gives $\mathrm{PoA}\le 2$.
\end{proof}

\begin{proof}[Proof of ~\Cref{prop:tightness}]
Consider $K$ builders and $K$ regions $r_0, r_1,\dots,r_{K-1}$. For each $i=0,\dots,K-1$, let source $I_i$ be located in region $r_i$ and suppose that it is received only by builders located in region $r_i$. That is, for all $t\in[0,\Delta]$,
\[
q_{r(I_i)}(r_i,t)=1,
\qquad
q_{r(I_i)}(r,t)=0
\quad\text{for all } r\neq r_i.
\]

Assign expected emitted values
\[
\Lambda_{{I_0}} = K,
\qquad
\Lambda_{I_i}=1-\varepsilon
\qquad 
\text{for }
i = 1,\dots,K-1.
\]

Consider the profile $\mathbf{s}$ in which all $K$ builders choose region $r_0$, i.e., $s_b=r_0$ for all $b\in\mathcal{B}$. Under the equal-split sharing rule, every transaction from source $I_0$ is received by all $K$ builders and shared equally among them. Since the expected emitted value of $I_0$ is $K$, each builder's expected utility is
\[
u_b(\mathbf{s})=\frac{\Lambda_{I_0}}{K}=\frac{K}{K}=1.
\]

If builder $b$ deviates to some region $r_i$ with $i\geq 1$, then it receives only transactions from source $I_i$. Its expected utility from the deviation is therefore
\[
u_b(r_i,\mathbf{s}_{-b})=\Lambda_{I_i}=1-\varepsilon<1.
\]
Deviating to any secondary source's region yields the same payoff, which is lower than the payoff from remaining at $r_0$. Hence, $\mathbf{s}$ is a strict pure Nash equilibrium.

The welfare under profile $\mathbf{s}$ is 
\[
W(\mathbf{s}) = \Lambda_{I_0}=K.
\]

A welfare-maximizing profile $\mathbf{s}^*$ places one builder at each region $r_i$, so that all $K$ sources are covered. Its welfare is therefore
\[
W(\mathbf{s}^*) = \sum_{i=0}^{K-1}\,\Lambda_{I_i} = K+(1-\varepsilon)(K-1)
\]
Thus,
\[
\frac{W(\mathbf{s}^*)}{W(\mathbf{s})}
= 
\frac{K+(1-\varepsilon)(K-1)}{K}
=
2 - \frac{1}{K} - \frac{K-1}{K}\varepsilon.
\]
Letting $K \to +\infty$ and $\varepsilon \to 0^+$ shows that the \gls{poa} can be arbitrarily close to $2$.
\end{proof}

\added{
\begin{proof}[Proof of \Cref{remark:pos}]
It remains to show that $\mathbf{s}$ is the unique PNE. Consider any profile with a builder in some secondary region $r_i, i \geq1$. Let $m_i \geq 1$ be the number of builders in the secondary region $r_i$, and $m_0 \leq K-1$ be the number of builders in $r_0$. A builder in $r_i$ obtains utility
\[
\frac{1-\varepsilon}{m_i}\leq 1-\varepsilon,
\]
whereas moving to $r_0$ yields 
\[
\frac{K}{(m_0+1)}\geq1.
\]
The builder in $r_i$ has a strictly profitable deviation. Thus, every PNE places all builders in $r_0$, proving uniqueness. Since the unique PNE has the welfare ratio above, $\mathrm{PoS} \leq \mathrm{PoA} \leq 2.$
\end{proof}
} 

\begin{proof}[Proof of \Cref{lem:utility_dispersion}]
Let
\[
\underline b \in \arg\min_{b\in\mathcal{B}} u_b(\mathbf{s}),
\qquad
\bar b \in \arg\max_{b\in\mathcal{B}} u_b(\mathbf{s}).
\]
Since $\mathbf{s}$ is a Nash equilibrium, builder $\underline b$ cannot improve its utility by moving to the region chosen by builder ${\bar b}$. Hence,
\[
u_{\underline b}(\mathbf{s})
\geq
u_{\underline b}(s_{\bar b},\mathbf{s}_{-\underline b}).
\]
It is therefore enough to prove that
\[
u_{\underline b}(s_{\bar b},\mathbf{s}_{-\underline b})
\geq
\frac{1}{2}u_{\bar b}(\mathbf{s}).
\]

Fix a realized transaction $j$ with value $V_j$, emitted by source $I$ at time $t_j$. Condition on the reception outcomes of all builders in $\mathcal{B}\setminus\{\underline b,\bar b\}$. Let
\[
m
:=
\left|\{c\in\mathcal{B}\setminus\{\underline b,\bar b\}: j\in\mathcal{T}_c\}\right|
\]
be the number of builders other than $\underline b$ and $\bar b$ that receive $j$, and write
\[
q_j(s_{\bar b}):=q_{r(I)}(s_{\bar b},t_j).
\]

% Let $x$ be the number of those builders that receive $j$, and let
% \begin{equation*}
% q_j := q_{I(j)}(s_{\bar b}, t_j).
% \end{equation*}
Under the original profile $\mathbf{s}$, builder $\bar b$ receives $j$ with probability $q_j(s_{\bar b})$. Conditional on receiving it, at least $m$ other builders also receive $j$. Therefore, builder $\bar b$'s share is at most $1/(m+1)$, and its conditional expected reward from $j$ is at most
\[
V_j \frac{q_j(s_{\bar b})}{m+1}.
\]

Now suppose builder $\underline b$ deviates to region $s_{\bar b}$. Then builder $\underline b$ receives $j$ with the same probability $q_j(s_{\bar b})$. Conditional on receiving it, at most $m+1$ other builders can also receive $j$: the $m$ fixed builders and possibly builder $\bar b$. Therefore, builder $\underline b$'s share is at least $1/(m+2)$, and its conditional expected reward from $j$ is at least
\[
V_j \frac{q_j(s_{\bar b})}{m+2}.
\]
Since
\[
V_j \frac{q_j(s_{\bar b})}{m+2}
\geq
\frac{1}{2}V_j \frac{q_j(s_{\bar b})}{m+1},
\]
the deviating builder's conditional expected reward from transaction $j$ is at least half of builder $\bar b$'s conditional expected reward from $j$.

Summing over all realized transactions and taking expectations gives
\[
u_{\underline b}(s_{\bar b}, \mathbf{s}_{-\underline b})
\geq
\frac{1}{2} u_{\bar b}(\mathbf{s}).
\]

Combining this with the Nash equilibrium inequality yields
\[
u_{\min}(\mathbf{s})
=
u_{\underline b}(\mathbf{s})
\geq
\frac{1}{2} u_{\max}(\mathbf{s}).
\]
\end{proof}

\begin{proof}[Proof of \Cref{prop:utility_dispersion_tight}]
Consider two builders $b_1,b_2\in\mathcal{B}$, two regions $r_1,r_2$, and two sources $I_1,I_2$, where source $I_i$ is located in region $r_i$. Suppose reception is deterministic and exclusive: for all $t\in[0,\Delta]$ and $i\in\{1,2\}$,
\[
q_{r(I_i)}(r_i,t)=1,
\qquad
q_{r(I_i)}(r,t)=0
\quad
\text{for all } r\neq\,r_i.
\]

Let
\[
\Lambda_{I_1}=2,
\qquad
\Lambda_{I_2}=1,
\]
and consider the profile
\[
\mathbf{s}=(s_{b_1}=r_1,\;s_{b_2}=r_2).
\]
Then
\[
u_{b_1}(\mathbf{s})=2,
\qquad
u_{b_2}(\mathbf{s})=1.
\]

If $b_2$ deviates to $r_1$, both builders receive only transactions from source $I_1$ and split its value equally. Thus
\[
u_{b_2}(r_1,\mathbf{s}_{-b_2})
=
\frac{\Lambda_{I_1}}{2}
=
1
=
u_{b_2}(\mathbf{s}).
\]

If $b_1$ deviates to $r_2$, both builders receive only transactions from source $I_2$ and split its value equally. Thus 
\[
u_{b_1}(r_2,\mathbf{s}_{-b_1})
=
\frac{\Lambda_{I_2}}{2}
=
\frac{1}{2}
<
2 = u_{b_1}(\mathbf{s}).
\]
Hence $\mathbf{s}$ is a pure Nash equilibrium. Furthermore,
\[
u_{\min}(\mathbf{s})=1=\frac{1}{2}\cdot 2=\frac{1}{2}u_{\max}(\mathbf{s}).
\]
Therefore, the factor $1/2$ is tight.
\end{proof}

\begin{proof}[Proof of \Cref{prop:utility_hhi}]
Let $x_b := x_b(\mathbf{s})$ and let $a := \min_{b \in \mathcal{B}} x_b.$ By \Cref{lem:utility_dispersion}, $u_{\min}(\mathbf{s}) \ge \frac{1}{2}u_{\max}(\mathbf{s})$, so every utility share satisfies $x_b \in [a,2a].$ Write
\[
x_b = a + y_b,
\qquad
0 \le y_b \le a. 
\]
Since $\sum_{b \in \mathcal{B}} x_b = 1$, we have
\[
\sum_{b \in \mathcal{B}} y_b = 1 - Ka.
\]
Therefore,
\begin{align*}
\mathrm{HHI}(\mathbf{s})
&= \sum_{b \in \mathcal{B}} (a+y_b)^2 \\
&= Ka^2 + 2a \sum_{b \in \mathcal{B}} y_b + \sum_{b \in \mathcal{B}} y_b^2 \\
&\le Ka^2 + 2a(1-Ka) + a(1-Ka) \\
&= 3a - 2Ka^2,
\end{align*}
where we used $y_b^2 \le ay_b$. Since each $x_b \in [a,2a]$ and $\sum_{b \in \mathcal{B}} x_b = 1$, we also have
\[
Ka \le 1 \le 2Ka,
\]
so $a \in [1/(2K),1/K]$. The quadratic $3a - 2Ka^2$ is maximized at $a = 3/(4K)$, which lies in the interval. Hence,
\[
\mathrm{HHI}(\mathbf{s}) \le \frac{9}{8K}.
\]

To see tightness, let $K = 3\ell$. Consider $K$ regions and $K$ sources, each receivable only in its own region with deterministic binary reception. Let $\ell$ sources have expected value $2$, and let the remaining $2\ell$ sources have expected value $1$. Place one builder in each region.

This profile is a pure Nash equilibrium. A builder located at a value-$1$ source is indifferent to moving to a value-$2$ source, where it would split that source with the incumbent builder and obtain payoff $1$. Moving to any other region yields at most $1/2$ payoff. A builder located at a value-$2$ source obtains at most $1$ by deviating, thus cannot improve its payoff. Hence, the profile is a pure Nash equilibrium.

The resulting utilities are $\ell$ copies of $2$ and $2\ell$ copies of $1$. Since total utility is
\[
2\ell + 2\ell = 4\ell = \frac{4K}{3},
\]
the normalized utility shares are $\ell$ copies of $3/(2K)$ and $2\ell$ copies of $3/(4K)$. Hence
\[
\mathrm{HHI}(\mathbf{s})
= \frac{K}{3}\left(\frac{3}{2K}\right)^2
+ \frac{2K}{3}\left(\frac{3}{4K}\right)^2
= \frac{9}{8K}.
\]
Therefore, the factor $9/8$ is tight as a universal bound.
\end{proof}

    \section{List of Google Cloud Platform Regions}
\label{app:gcp-regions}

\begin{table}[H]
\centering
\small
% \resizebox{\linewidth}{!}{
\begin{tabular}{lll}
\toprule
\textbf{Type} & \textbf{GCP region} & \textbf{Location} \\
\midrule
\multirow{11}{*}{High-value source pool}
& us-east1 & Moncks Corner, SC, USA \\
& us-east4 & Ashburn, Virginia, USA \\
& us-central1 & Council Bluffs, Iowa, USA \\
& europe-west1 & St. Ghislain, Belgium \\
& europe-west2 & London, UK \\
& europe-west3 & Frankfurt, Germany \\
& europe-west4 & Eemshaven, Netherlands \\
& europe-north1 & Hamina, Finland \\
& asia-northeast1 & Tokyo, Japan \\
& asia-northeast2 & Osaka, Japan \\
& asia-southeast1 & Singapore \\
\midrule
\multirow{9}{*}{Peripheral source pool}
& us-west1 & The Dalles, Oregon, USA \\
& us-west2 & Los Angeles, California, USA \\
& asia-south1 & Mumbai, India \\
& asia-south2 & Delhi, India \\
& australia-southeast1 & Sydney, Australia \\
& australia-southeast2 & Melbourne, Australia \\
& southamerica-east1 & Sao Paulo, Brazil \\
& southamerica-west1 & Santiago, Chile \\
& africa-south1 & Johannesburg, South Africa \\
\midrule
\multirow{4}{*}{Placement only / intermediate}
& europe-central2 & Warsaw, Poland \\
& europe-southwest1 & Madrid, Spain \\
& me-central1 & Doha, Qatar \\
& me-west1 & Tel Aviv, Israel \\
\bottomrule
\end{tabular}
\caption{Candidate GCP regions and their source-pool assignments.}
\label{tab:gcp-region-pools}
% }
\end{table}

\begin{figure}[H]
    \centering
    \includegraphics[width=0.9\linewidth]{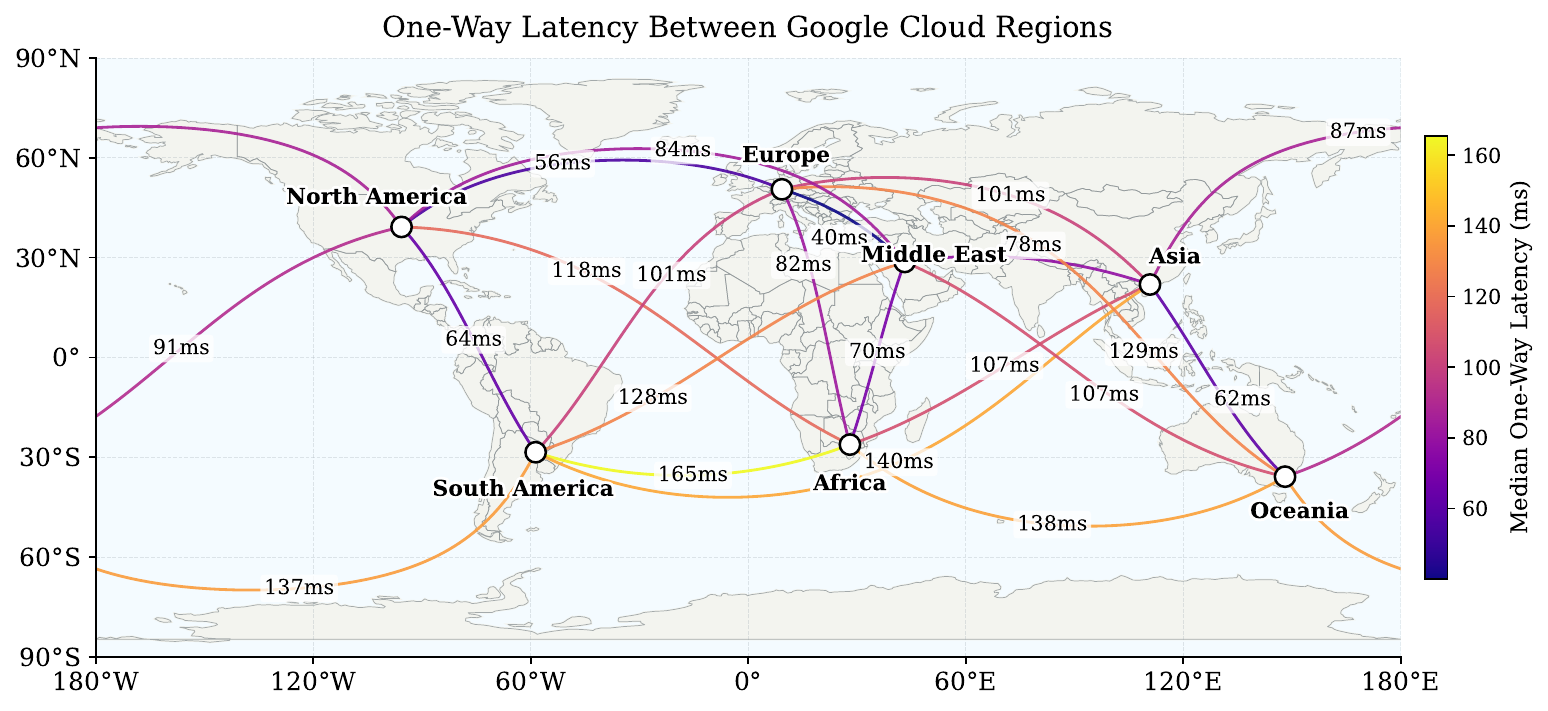}
    \caption{One-way latency between continents. Nodes denote continents obtained by grouping relevant GCP regions. Curves show median one-way latency between continent pairs.}
    \label{fig:latency_map}
\end{figure} 

All experiments use a fixed set of 24 candidate regions in \Cref{tab:gcp-region-pools}, selected from the empirical GCP latency matrix to span North America, Europe, Asia, Australia, South America, Africa, and the Middle East.
For heterogeneous information-source layouts, we distinguish between two source pools: a high-value pool, consisting of regions near major financial or technology centers, and a peripheral pool of regions geographically more remote from these centers.
In each randomized source-layout instance, high-value sources are sampled from the high-value pool and peripheral sources are sampled from the peripheral pool.
Regions outside these pools remain available as builder locations but are not used as source locations in the two-cluster construction. \Cref{fig:latency_map} presents the one-way latency between continents, obtained by grouping GCP regions into their corresponding continents and taking the median one-way latency across all inter-continent GCP region pairs.

\section{Greedy Planner Benchmark Validation}
\label{app:greedy-validation}

For experiments with larger builder-count sweeps, computing the exact planner benchmark by exhaustive search becomes expensive. With $|\mathcal R|=24$ candidate regions and builder co-location allowed, the number of feasible $K$-builder placement multisets is
$\binom{|\mathcal R|+K-1}{K}$.
This number is already $475{,}020$ for $K=6$ and grows to $834{,}451{,}800$ for $K=12$, the largest builder count in our sweep.
We therefore use a greedy welfare-maximization routine as the planner benchmark in the large-$K$ experiments.

To check whether this choice materially affects our reported welfare ratios, we compare the greedy planner benchmark against exhaustive search on tractable instances, covering $K=3$ through $K=7$.
For each tested instance, let $\widehat W^*$ be the planner welfare returned by the greedy routine and let $W^*$ be the exact planner welfare found by exhaustive search.
We report the ratio $\widehat W^*/W^*$ to measure how close the greedy planner benchmark is to the exact planner optimum.

As shown in \Cref{tab:greedy-brute-validation}, the greedy planner matches the exact planner in median for every tested value of $K$.
Across all 25 tested instances, the worst observed greedy-to-exact welfare ratio is $0.999129$, with a maximum welfare gap of $0.000871$ and a maximum induced welfare-ratio overstatement of $0.000872$.
Thus, on tractable instances, the greedy benchmark introduces only a negligible difference relative to exhaustive search.

\begin{table}[H]
\centering
\small
\begin{tabular}{lllll}
\toprule
$K$ & Profiles & Instances & Median $\widehat W^*/W^*$ & Min $\widehat W^*/W^*$ \\
\midrule
3 & 2,600 & 5 & 1.000000 & 1.000000  \\
4 & 17,550 & 5 & 1.000000 & 0.999915 \\
5 & 98,280 & 5 & 1.000000 & 0.999129 \\
6 & 475,020 & 5 & 1.000000 & 0.999390 \\
7 & 2,035,800 & 5 & 1.000000 & 0.999851  \\
\midrule
All & -- & 25 & 1.000000 & 0.999129 \\
\bottomrule
\end{tabular}
\caption{Comparison between the greedy planner benchmark and the exhaustive-search planner benchmark on tractable instances.}
\label{tab:greedy-brute-validation}
\end{table}

\section{Experiments~2 and~4 with lower source-value asymmetry}
\label{app:exp2_low_gamma}
% \label{app:exp4_low_gamma}

\parhead{Experiment~2}
\Cref{fig:exp2_gamma2} repeats the slot-duration experiment with a lower high-to-low value ratio, $\gamma=2$. In this setting, peripheral sources account for a larger share of total emitted value.

The qualitative patterns match \Cref{fig:exp2}, with three differences. The median welfare ratio drops to $0.885$ at $\Delta=100$ms, compared to $0.97$ for $\gamma=10$. Peripheral sources carry a larger share of the total value, so missing their coverage in equilibrium has a larger welfare impact. Peripheral coverage in the planner benchmark also rises at shorter $\Delta$, since each peripheral source is now worth a larger share of the total value. Finally, planner utility HHI rises above the $1/K$ baseline at low slot durations since the planner allocates some builders to peripheral regions, where these builders earn less than builders in high-value regions. In equilibrium, builders remain in high-value regions and utility HHI stays at the floor. Geographic HHI is essentially unchanged from the $\gamma=10$ case.

\begin{figure}[H]
\centering
\includegraphics[width=\linewidth]{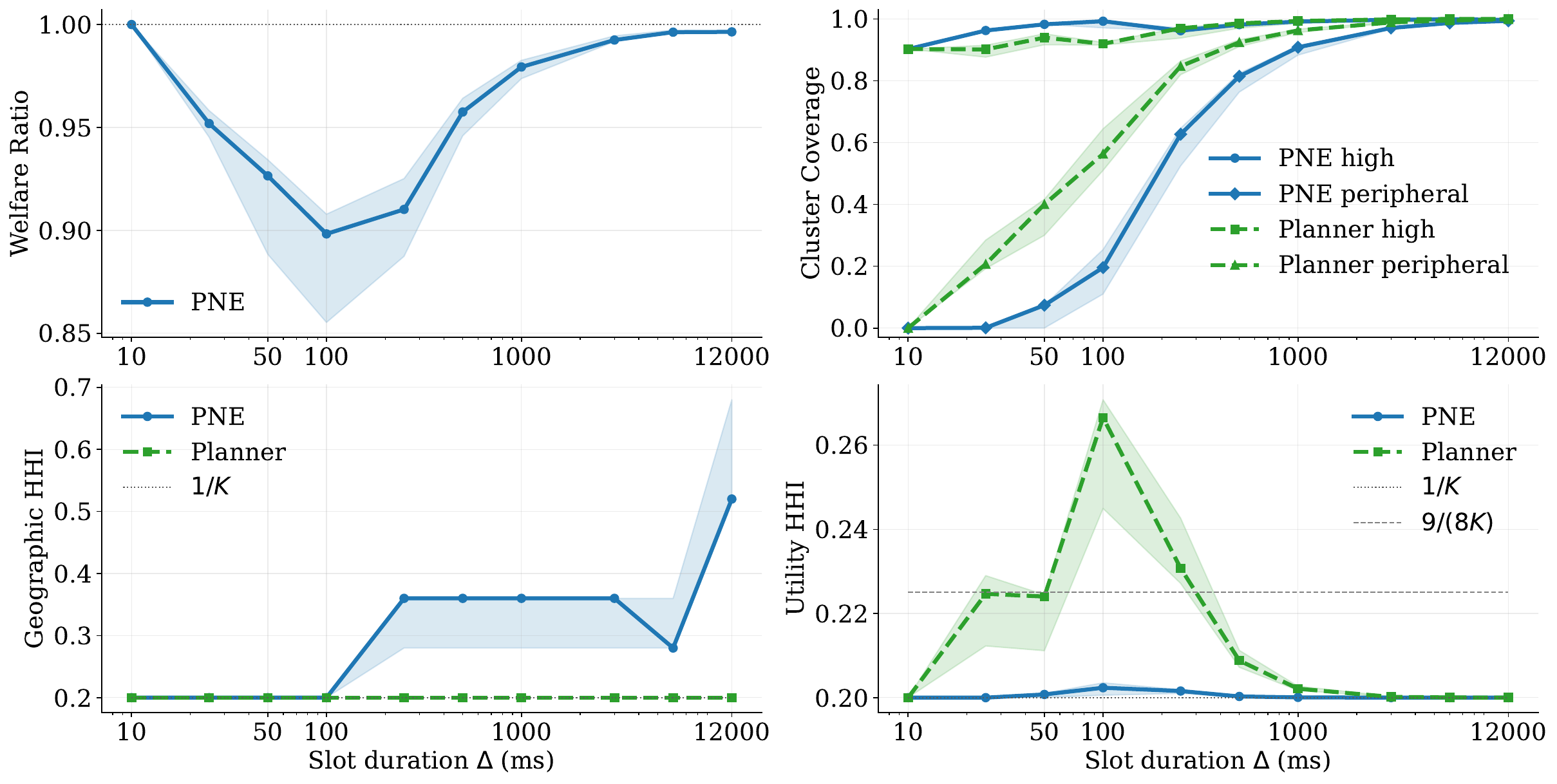}
\caption{\textit{Effect of slot duration at lower source-value asymmetry ($\gamma=2$).} Same setup as \Cref{fig:exp2} but with reduced asymmetry.}
\label{fig:exp2_gamma2}
\end{figure}

\parhead{Experiment~4}
\Cref{fig:exp4_gamma2} repeats the builder-participation experiment with a lower high-to-low value ratio, $\gamma=2$.

The qualitative patterns largely match those in \Cref{fig:exp4}, but the welfare gap is more pronounced. For small $K$, additional builders mainly improve high-value coverage. As $K$ grows, builders are also incentivized to cover peripheral sources, unlike in the higher-asymmetry case, where peripheral coverage remains zero for much of the sweep. However, this peripheral coverage still lags behind the planner, which expands peripheral coverage more aggressively. This explains both the larger welfare loss in the intermediate range of $K$ and partial recovery at larger $K$: once enough builders participate, equilibrium coverage begins to extend beyond the high-value cluster.

The concentration metrics show the same tradeoff as in the main experiment. For larger $K$, equilibrium builder placement remains more geographically centralized than the planner benchmark, while the planner stays at the egalitarian baseline $1/K$. Utility becomes more balanced for both profiles as $K$ grows, but it can remain slightly more unequal in the planner benchmark because the planner assigns more builders to the peripheral regions. Thus, lowering the value asymmetry makes peripheral coverage more important for welfare and widens the gap between equilibrium welfare and the optimum.

% First, the welfare ratio is lower over much of the intermediate-$K$ range.
% With lower source-value asymmetry, peripheral sources account for a larger share of the total value, so the welfare loss from missing peripheral coverage is larger.
% Second, unlike the $\gamma=10$ case, the ABR-selected PNE begins to cover peripheral sources once $K$ is moderately large, although this coverage still lags behind the planner benchmark.
% This explains why the welfare ratio partially recovers at larger $K$: additional builders allow the equilibrium outcome to cover both high-value and peripheral source clusters more effectively.
% Third, the geographic concentration pattern is broadly unchanged: for larger $K$, the outcome of ABR-selected PNE remains more geographically concentrated than the planner benchmark.
% The utility HHI also follows the same decreasing trend as $K$ grows, with the planner remaining slightly more unequal in utility because it assigns some builders to lower-value peripheral coverage positions.

\begin{figure}[H]
\centering
\includegraphics[width=\linewidth]{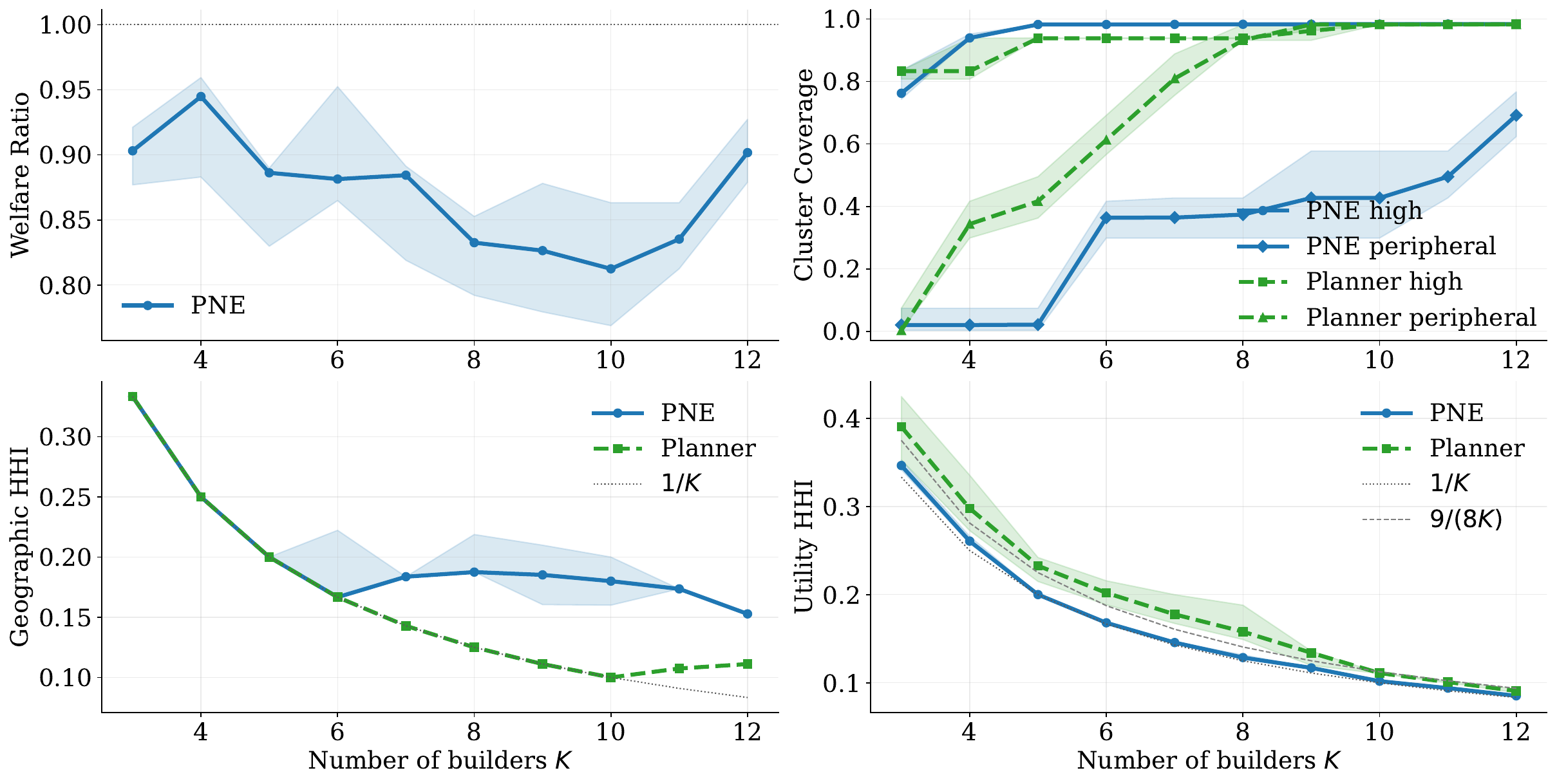}
\caption{\textit{Effect of builder participation at lower source-value asymmetry ($\gamma=2$).} Same setup as \Cref{fig:exp4} but with reduced asymmetry.}
\label{fig:exp4_gamma2}
\end{figure}

\section{Visualizations of PNE--Planner Builder Placement Difference}
\label{app:placement-view}

The welfare, coverage, and concentration metrics present different outcomes in equilibrium and in the planner benchmark, but they do not directly identify the underlying builder placements that drive these differences. To make this comparison more concrete, we visualize builder placements from one representative instance of the simulation experiments.

% \Cref{fig:exp1-allocation} corresponds to Experiment~1, where we vary the high-to-low source value ratio $\gamma$.
% \Cref{fig:exp2-allocation} corresponds to Experiment~2, where we vary the slot duration $\Delta$.
% \Cref{fig:exp4-allocation} corresponds to Experiment~4, where we vary the number of builders $K$.

In each figure, rows correspond to GCP regions and are grouped into high-value source regions, peripheral source regions, and regions without sources.
For each parameter setting, we show both the builder placement in equilibrium and the planner benchmark. The number inside each circle reports the number of builders assigned to that region.

\parhead{Experiment~1: source-value asymmetry}
\Cref{fig:exp1-allocation} visualizes the placement pattern behind Experiment~1.
At low value ratios, the planner assigns some builders to peripheral source regions because peripheral coverage still contributes substantially to welfare. In equilibrium, builders are incentivized to place themselves near high-value regions to access more valuable transactions.
As the value ratio $\gamma$ increases, high-value sources dominate the welfare objective, and the planner also shifts builder placement toward high-value source regions. Thus, the placement difference is most visible at low-to-intermediate value ratios, where peripheral coverage remains useful for the planner but is less attractive to self-interested builders in equilibrium.

% When $\gamma$ is small, the planner assigns some builders to peripheral source regions because peripheral coverage still contributes to welfare.
% The ABR-selected PNE shifts toward high-value source regions earlier, reflecting individual incentives to locate near higher-value sources.
% As $\gamma$ increases, high-value sources dominate the welfare objective, and the planner also concentrates builders in high-value source regions. Thus, the placement gap is clearest at low to intermediate $\gamma$, where peripheral coverage remains useful for the planner but is less attractive to individual builders.

\begin{figure}[htbp]
    \centering
    \includegraphics[width=\linewidth]{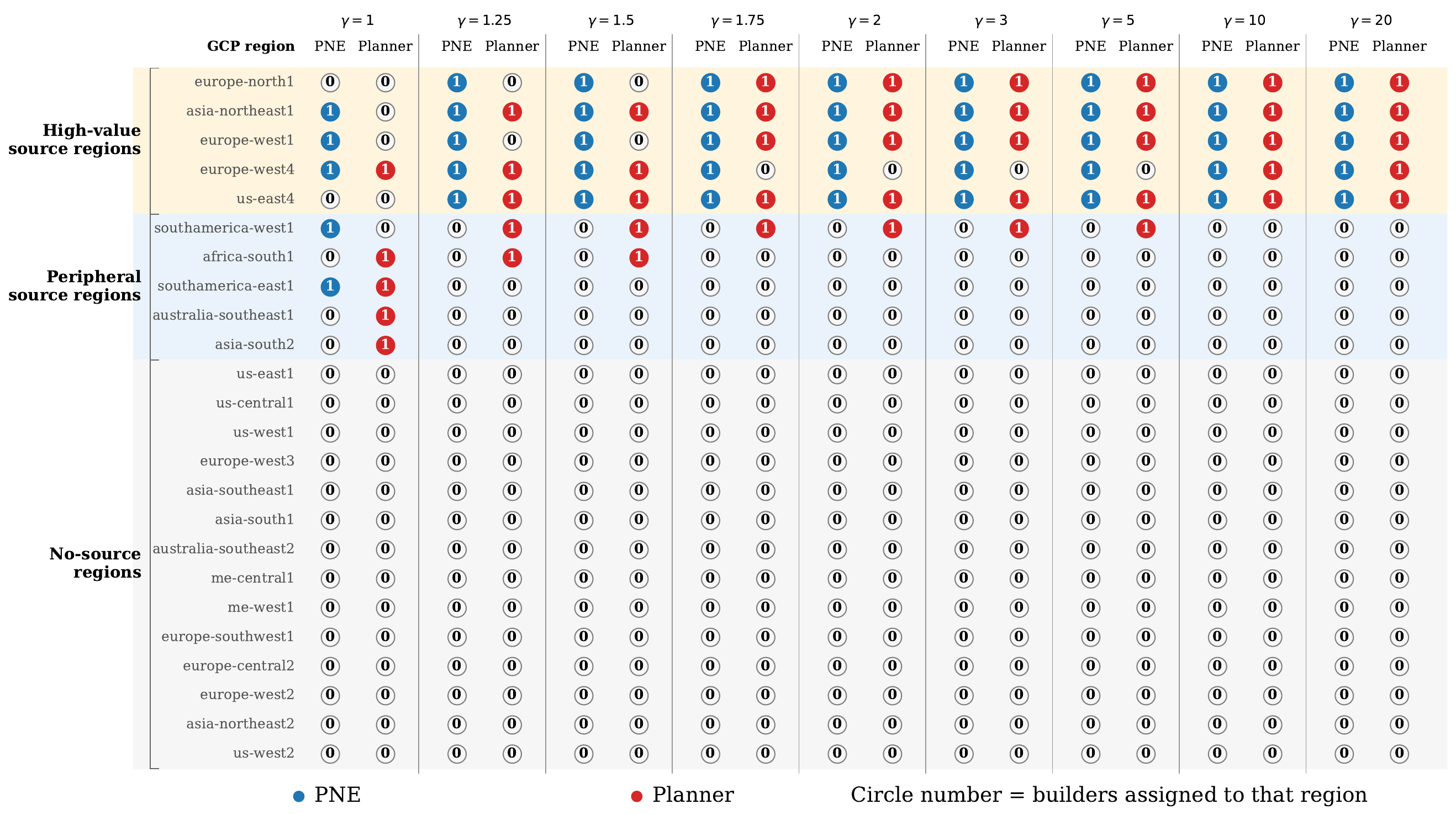}
    \caption{Builder placement by GCP region in equilibrium and the planner benchmark under varying high-to-low source value ratio $\gamma$, with fixed $K=5$ builders and slot duration $\Delta=50$ms. When source-value asymmetry starts to appear, builders are incentivized to locate in high-value source regions in equilibrium, while the planner still assigns builders to both high-value and peripheral source regions. As $\gamma$ increases, builder placement shifts toward high-value source regions in both profiles.}
    \label{fig:exp1-allocation}
\end{figure}

\parhead{Experiment~2: slot duration}
\Cref{fig:exp2-allocation} visualizes the placement pattern behind Experiment~2.
When the slot duration is short, source--builder proximity is critical: builder placements are near high-value source regions in both the equilibrium and the planner benchmark, since distant sources are difficult to cover before the deadline. At intermediate slot durations, peripheral sources become reachable from some regions, but coverage remains placement-sensitive. The planner spreads some builders toward peripheral source regions, while builders concentrate around high-value sources in equilibrium. In this regime, redundant high-value coverage can come at the expense of peripheral coverage. For long slot durations, geography matters less for coverage, and in both profiles, builders may locate in non-source regions that provide broad access to multiple source regions.

% When $\Delta$ is small, both outcomes place builders near high-value source regions because source proximity strongly affects welfare under tight propagation deadlines.
% As $\Delta$ grows, the planner starts assigning builders to peripheral source regions and, for long slots, to no-source regions that improve broader propagation coverage.
% At intermediate slot durations, peripheral sources become reachable, but coverage still depends strongly on where builders are placed. 
% The planner spreads some builders to peripheral source regions, while the ABR-selected PNE remains more concentrated around high-value source regions, leading to redundant high-value coverage and lower peripheral coverage.
% For long slot durations, geography matters less for transaction coverage, and both outcomes may place builders in non-source regions that provide broader propagation coverage.

\begin{figure}[htbp]
    \centering
    \includegraphics[width=\linewidth]{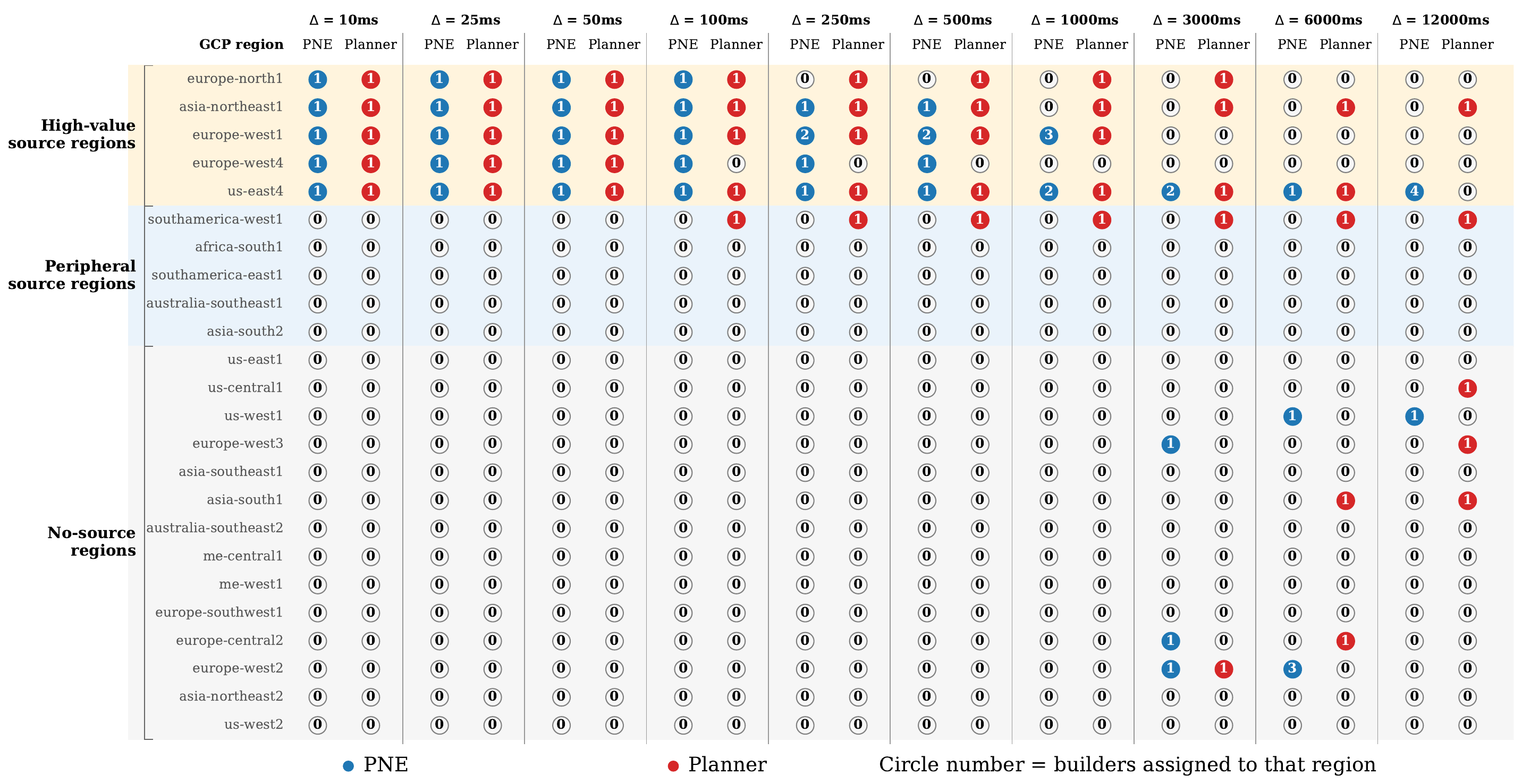}
    \caption{Builder placement by GCP region in equilibrium and the planner benchmark under varying slot duration $\Delta$, with the builder count fixed at $K=5$. For short slot durations, both profiles place most builders in high-value source regions, since source proximity is critical when propagation time is tight. As the slot duration increases, the planner starts assigning builders to peripheral source regions, while builders remain more concentrated on high-value regions in equilibrium. At long slot durations, both profiles place some builders in non-source GCP regions that provide broader propagation coverage.}
    \label{fig:exp2-allocation}
\end{figure}

\parhead{Experiment~3: joint variation of source-value asymmetry and slot duration}
\Cref{fig:exp3-allocation} shows representative placements from Experiment~3, where both the value ratio $\gamma$ and slot duration $\Delta$ vary. We show a small set of parameter settings rather than the full grid. The examples illustrate why neither value asymmetry nor slot duration alone determines placement difference.

When slots are short and value asymmetry is high, both profiles concentrate on high-value source regions, so welfare is similar. At intermediate slot durations and moderate value ratios, peripheral sources are reachable and still valuable.
In this regime, the planner assigns builders to preserve peripheral coverage, while builders remain more concentrated around high-value source access in equilibrium. Increasing $\gamma$ reduces the welfare relevance of peripheral coverage, so the welfare penalty of concentrated placement in equilibrium becomes lower. At long slot durations, coverage becomes less sensitive to exact placement, and both profiles can maintain broader coverage. These examples illustrate the main mechanism behind the welfare gap: peripheral coverage must be feasible and valuable, while individual builder incentives still favor high-value source regions.

% corresponds to Experiment~3, where we study the joint effect of the high-to-low source value ratio $\gamma$ and the slot duration $\Delta$. Because this experiment varies two parameters jointly, we show four representative $(\gamma,\Delta)$ settings rather than the full grid, covering short to long slot durations and moderate to high source-value asymmetry.

% The four examples in \Cref{fig:exp3-allocation} show that neither $\gamma$ nor $\Delta$ alone determines the placement gap.
% With a short slot duration and high source-value asymmetry, both outcomes concentrate on high-value sources, yielding similar welfare.
% With an intermediate slot duration and moderate asymmetry, peripheral sources are reachable and still valuable, so the planner preserves peripheral coverage while the ABR-selected PNE outcome remains concentrated on high-value sources.
% Increasing $\gamma$ in the intermediate-slot regime reduces this gap because high-value coverage becomes more important for both outcomes.
% With a long slot duration and moderate asymmetry, coverage becomes less sensitive to placement, and both outcomes can maintain a broad presence across the two source clusters.
% These examples show that the welfare gap arises most clearly when peripheral coverage is both feasible and valuable, while individual builder incentives still favor high-value source regions.

\begin{figure}
    \centering
    \includegraphics[width=\linewidth]{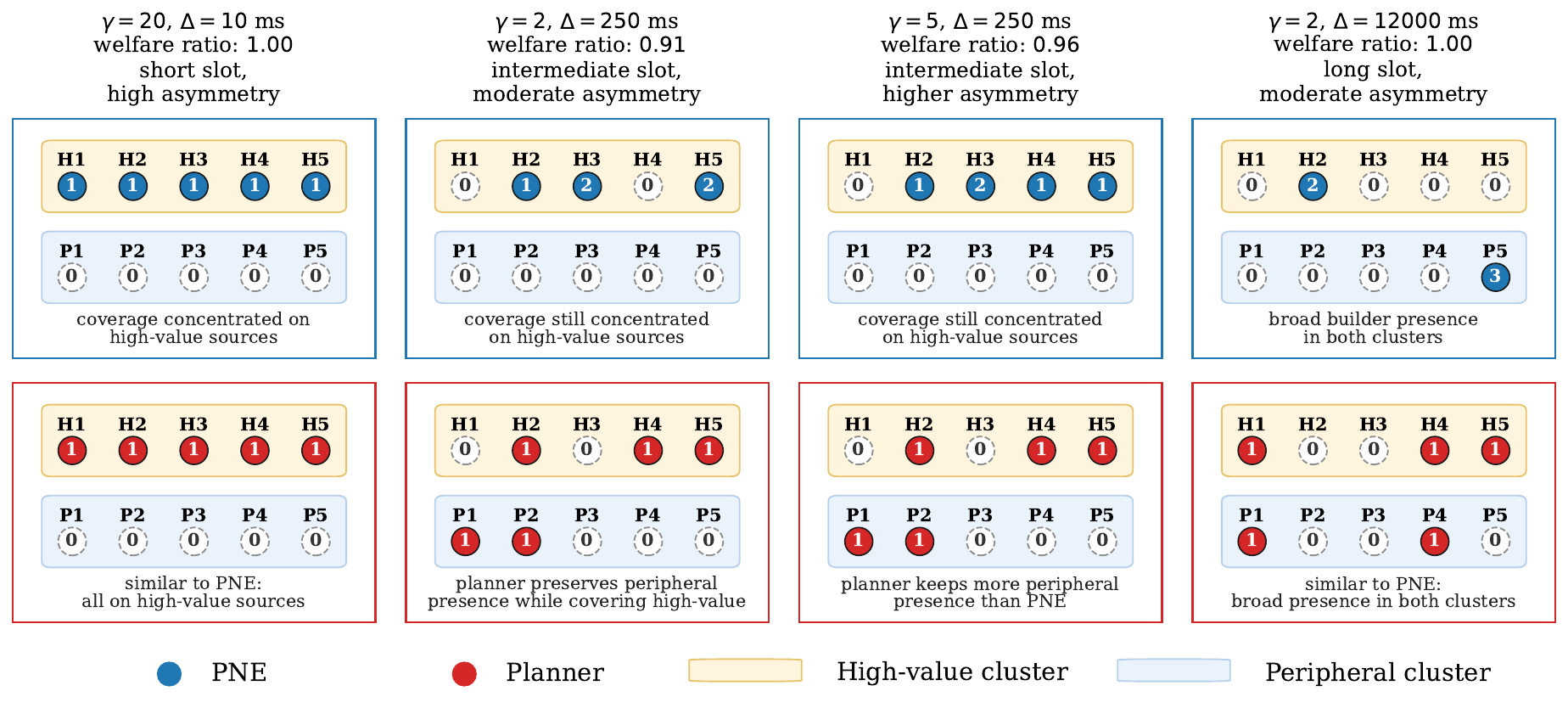}
    \caption{Representative placement instances for four parameter settings in Experiment~3. In each panel, the top row shows the builder placement in equilibrium, while the bottom row shows the planner placement. Yellow boxes denote high-value source locations, and blue boxes denote peripheral source locations.
    Labels H1--H5 and P1--P5 are shorthand for the five GCP regions in the high-value source cluster and the five GCP regions in the peripheral source cluster, respectively.
    The number inside each circle reports the number of builders assigned to that location.
    These examples show that the welfare gap can arise when peripheral sources contribute to meaningful value and are reachable from certain regions but are uncovered in equilibrium.
    }
    \label{fig:exp3-allocation}
\end{figure}

\parhead{Experiment~4: builder participation}
\Cref{fig:exp4-allocation} visualizes the placement pattern behind Experiment~4. Because this experiment uses a greedy planner benchmark rather than the exact exhaustive search, the planner's builder placement should be interpreted as the greedy benchmark profile. As the number of builders $K$ increases, builders remain strongly incentivized to place themselves in regions with better access to high-value sources. Once high-value coverage is nearly saturated, additional builders can therefore create redundant coverage. The planner instead allocates some of the additional builders to expand peripheral coverage. This placement difference explains why additional builder participation does not automatically produce broad source coverage in equilibrium.
% shows the regional placement pattern as the number of builders $K$ varies.
% Since this experiment reports the greedy planner benchmark rather than the exact planner benchmark, the planner placement pattern differs slightly from the previous figures when $K=5$.
% Nevertheless, the main trend remains clear.
% As $K$ increases, the ABR-selected PNE increasingly co-locates builders in high-value source regions, while the planner assigns additional builders to peripheral source regions once high-value coverage is largely saturated.
% This placement difference explains why the welfare gap grows for larger builder counts: the extra builders improve coverage under the planner benchmark, but under the ABR-selected PNE, they can produce redundant coverage around already attractive high-value regions.

\begin{figure}[htbp]
    \centering
    \includegraphics[width=\linewidth]{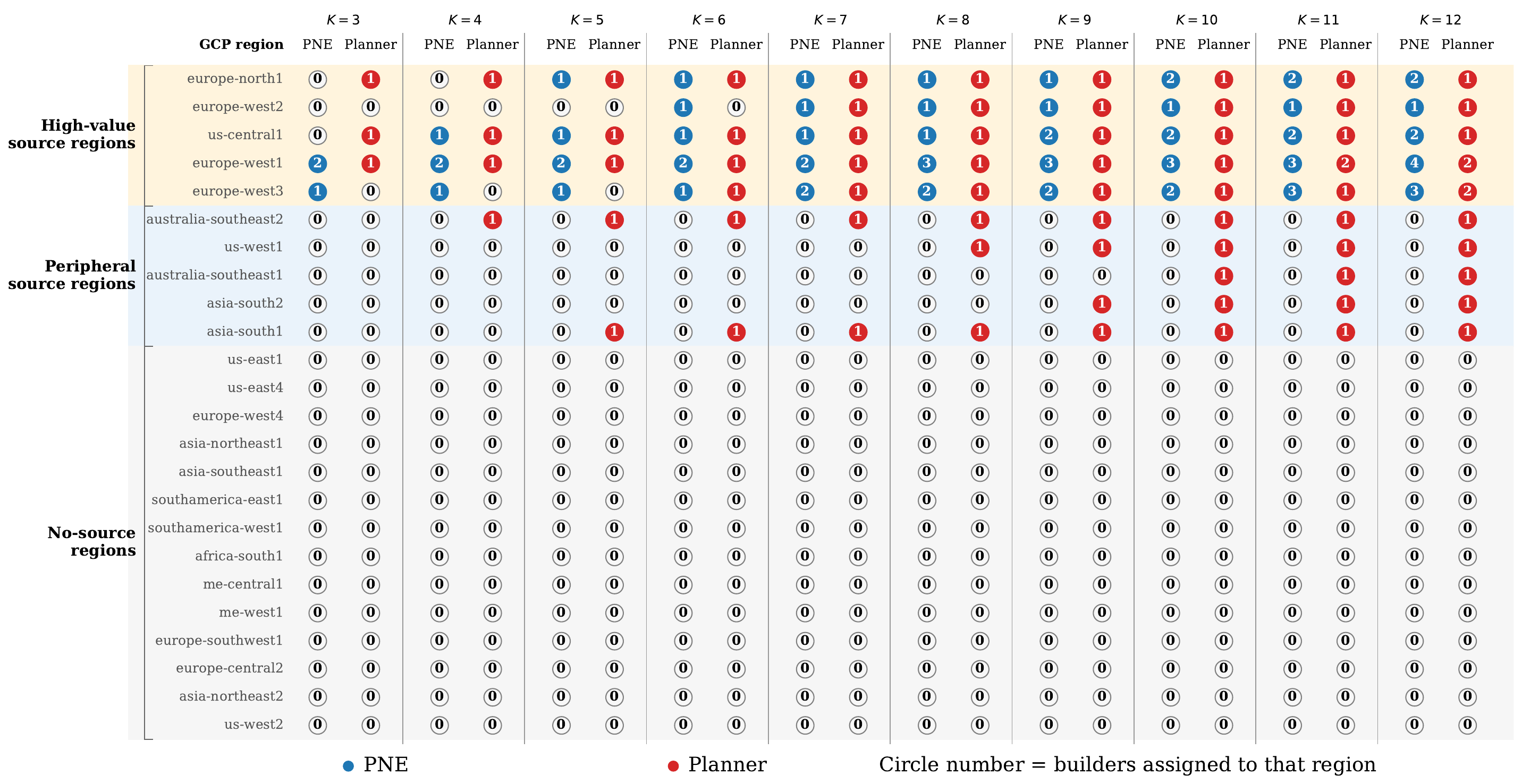}
    \caption{Builder placement by GCP region in equilibrium and the planner benchmark as the number of builders $K$ varies from 3 to 12. Builders are strongly incentivized to locate in high-value source regions in equilibrium, while the planner assigns additional builders to peripheral source regions once high-value coverage is largely saturated.}
    \label{fig:exp4-allocation}
\end{figure}
\fi

\end{document}